\documentclass{article}
\usepackage{amssymb,amsmath,amsthm, graphicx, esint}
\usepackage{tikz-cd}

\newcommand{\td}{\mathbf{d}}
\newcommand{\cZ}{\mathcal{Z}}
\newcommand{\bL}{\mathbb{L}}
\newcommand{\fg}{\mathfrak{g}}

\newcommand{\fk}{\mathfrak{k}}

\newcommand{\C}{\mathbb{C}}
\newcommand{\R}{\mathbb{R}}
\newcommand{\CP}{\mathbb{C}\mathrm{P}}
\newcommand{\Z}{\mathbb{Z}}
\newcommand{\N}{\mathbb{N}}
\newcommand{\one}{\mathbf{1}}
\newcommand{\SU}{\mathrm{SU}}
\newcommand{\U}{\mathrm{U}}
\newcommand{\bS}{\mathbb{S}}
\newcommand{\su}{\mathfrak{su}}
\newcommand{\fu}{\mathfrak{u}}

\newcommand{\cG}{\mathcal{G}}

\newcommand{\cH}{\mathcal{H}}
\newcommand{\Diff}{\mathrm{Diff}}
\newcommand{\SDiff}{\mathrm{SDiff}}
\newcommand{\X}{\mathfrak{X}}
\newcommand{\lra}[2]{\langle#1,#2\rangle}
\newcommand{\oline}{\overline}

\newcommand{\tr}{\mathrm{tr}}
\newcommand{\bra}[1]{\langle#1|}
\newcommand{\ket}[1]{|#1\rangle}
\newcommand{\vol}{\mathrm{Vol}}

\newtheorem{Theorem}{Theorem}
\newtheorem*{TheoremstarA}{Theorem A}
\newtheorem*{TheoremstarB}{Theorem B}

\newtheorem{Proposition}[Theorem]{Proposition}
\newtheorem{Lemma}[Theorem]{Lemma}
\newtheorem{Corollary}[Theorem]{Corollary}

\newtheorem{Definition}[Theorem]{Definition}

\theoremstyle{remark}
\newtheorem{Remark}{Remark}

\title{Central extensions for loop groups of area-preserving diffeomorphisms and their fuzzy sphere limits}
\author{Bas Janssens and Zhenghan Wang}
\begin{document}

\maketitle
\begin{abstract}
We classify central extensions for the loop group $L\SDiff(\bS^2)$ 
of area-preserving diffeomorphisms of the $2$-sphere, and of related twisted loop groups. 
We then show that the corresponding Lie algebra cocycles are `fuzzy sphere limits' of 
Kac-Moody cocycles for (twisted) loop algebras ${L\mathfrak{su}(k+1)}$ with $k\rightarrow \infty$, provided that the cocycles are rescaled by $6/k^3$. 
\end{abstract}

\section{Introduction}

Quantum field theories that are mathematically rigorous as well as experimentally relevant are interesting, important, and rare. 
One such class is the rational conformal field theories in $1+1$ dimension.  A crucial ingredient for success in this class is the 
symmetry enhancement of the group $\Diff(\bS^1)\times \Diff(\bS^1)$ of global conformal symmetries 
of $\bS^1 \times \bS^1$ to a central extension by $\mathrm{U}(1)$. 
It is the unitary representation theory of the centrally extended Lie group that plays a key role in the construction of the CFTs.

An extension of this approach to higher dimensions would be an interesting challenge. In this paper, we explore the possibility 
that in the construction of quantum field theories (QFTs) in  $2+1$ dimensions,
the loop group $L\SDiff(\bS^2)$ with values in the group of volume-preserving diffeomorphisms $\SDiff(\bS^2)$ of the 2-sphere -- considered as diffeomorphisms of $\bS^1 \times \bS^2$ -- may play a role that is similar to that of the conformal group $\Diff(\bS^1)\times \Diff(\bS^1)$ 
for $1+1$ dimensional field theories on $\bS^1 \times \bS^1$.  

There are many mathematical formulations of quantum field theory, but physically relevant examples are daunting to construct.  Our goal is to construct non-trivial $(2+1)$-dimensional conformal field theories mathematically based on the representation theory of extended symmetry groups, with an eye towards physical relevant examples such as the 3D Ising model \cite{Eck26}.

One way to think about general QFTs is as a triple $(|0\rangle, \mathcal{L},\mathcal{S})$, where $|0\rangle$ is the vacuum or ground state, $\mathcal{L}$ the local algebra, and $\mathcal{S}$ the non-local symmetry algebra.  In our case we think of $\mathcal{S}$ 
as the group algebra of a certain infinite dimensional Lie group, and what we wish to generalize to higher dimensions is the chiral algebra/vertex operator algebra formulation of chiral $(1+1)$-dimensional conformal field theories, and the construction of full conformal field theories using the representation theory of chiral algebras.  A generalization of chiral algebra to higher dimensions is formulated without non-trivial examples in \cite{Bor01}.  

Another general framework for QFTs is algebraic quantum field theory (AQFT).  In $(1+1)$-dimensional CFTs, vertex operator algebra and local conformal nets are essentially the same under suitable regularity assumptions \cite{CKLW18}.  It is natural to expect a similar correspondence for higher dimensions with the right formulations of CFTs.  

In either framework, a natural starting point for constructing nontrivial QFTs on the conformal compactification 
$(\mathbb{S}^1\times \mathbb{S}^d)/\Z_2$ of Minkowski space $\mathbb{R}^{1,d}$
would be to look for subgroups of $\Diff(\bS^1 \times \bS^d)^{\Z_2}$ that are large enough to give rise to 
local algebras of observables (so preferably infinite-dimensional), yet small enough to admit for
nontrivial central extensions at the Lie algebra level -- something that seems to be a prerequisite 
for representations of positive energy in certain settings \cite{JN24, JN25}. 

Since the Lie algebra of $\mathrm{Diff}(\bS^1 \times \bS^d)$ has interesting central extensions  
only if $d=0$ \cite{JN25}, one needs a subgroup already for $d=1$. In that case the conformal group 
$\Diff(\bS^1)\times \Diff(\bS^1)$ is a natural candidate: it has nontrivial central extensions (essentially 2 copies of the Virasoro-Bott group),
and the representation theory of these central extensions is a crucial ingredient in the construction of rational conformal field theories 
on $\bS^1 \times \bS^1$. 
We aim to explore the possibility that for $d=2$, the loop group 
$L\SDiff(\bS^2)$ with values in the group $\SDiff(\bS^2)$ 
of volume-preserving diffeomorphisms of the 2-sphere 
may play a similar role in the construction of QFTs on $\bS^1 \times \bS^2$.

The present paper is a preliminary step in our program. We determine the central extensions of $L\SDiff(\bS^2)$, and 
 we show that the corresponding Lie algebra cocycles can be obtained as `fuzzy sphere limits' of 
 Kac-Moody cocycles on $L\su(k+1)$ for $k\rightarrow \infty$, provided that the latter are rescaled by $6/k^3$.
 Since the conformal compactification of $\R^{1,2}$ is the quotient $(\bS^1 \times \bS^2)/\Z_2$, 
 we extend these results to the twisted setting.
 There are potential connections to the fuzzy sphere approach to the 3d Ising model that we will explore in the future \cite{Eck26}.
 One might hope that there exists some kind of limit of $L\mathrm{SU}(k+1)$-theories in the limit 
 $k \rightarrow \infty$ that has $L\SDiff(\bS^2)$ as a symmetry group. Indeed an appropriate limit of $\mathrm{SU}(k+1)$ as $k \rightarrow \infty$ is  $\SDiff(\bS^2)$ \cite{Hoppe89}. We plan to return to these issues in the future.
 
Our motivation to focus on the particular group $L\SDiff(\bS^2)$ is largely mathematical. The homotopy type of 
$\Diff(\bS^1 \times \bS^2)$ has two component, one that comes from $L\Diff(\bS^2)$ and one that comes from 
$\Diff(\bS^2)$ \cite{Hat81}. These two groups do not appear to admit a nontrivial central extensions at the Lie algebra level, 
and neither does $\SDiff(\bS^2)$ \cite{JV16}. In contrast, the loop group $L\SDiff(\bS^2)$ offers a Lie algebra 
cocycle that is integrable to the group level, and admits an interpretation as limits of cocycles on $L\su(k+1)$ that 
connect to well-understood conformal field theories. If we identify the Lie algebra of $L\SDiff(\bS^2)$ with the Poisson 
algebra $C^{\infty}(\bS^1 \times \bS^2)_{0}$ of functions that integrate to zero on every sphere $\{t\}\times \bS^2$, 
then this Lie algebra cocycle 
is \emph{local}, in the sense that the Lie algebras of functions with support in disjoint open sets commute in the centrally extended Lie algebra.
This locality condition is important in the AQFT formulation of $(2+1)$-CFTs.  It would be interesting to understand the corresponding formulation of locality in the chiral algebra formalism for $(2+1)$-CFTs, which is dual to AQFT by Fourier transform.

\subsection{Main results and structure of the paper}

Without loss of generality, we may assume that the volume form $\mu$ on $\bS^2$ is $\mathrm{SO}(3)$-invariant. 
We consider the group $\SDiff(\bS^2, \mu)$ of volume-preserving diffeomorphisms of $\bS^2$ as a Fr\'echet--Lie group
with Lie algebra $C^{\infty}_0(\bS^2)$, the ideal in the Poisson algebra $C^{\infty}(\bS^2)$ that consists of functions that integrate to zero.

We denote by $L\SDiff(\bS^2, \mu)$ the loop group of smooth functions $\bS^1 \rightarrow \SDiff(\bS^2)$ with pointwise multiplication, and by 
$LC^{\infty}_0(\bS^2)$ the loop algebra of smooth functions $F \colon \bS^1 \rightarrow C^{\infty}_0(\bS^2)$, equipped with the pointwise Poisson bracket
\[
[F,G](t) := \{F(t), G(t)\}.
\]
We can interpret $F\in LC^{\infty}_{0}(\bS^2)$ as a smooth function $\bS^1 \times \bS^2 \rightarrow \R$, and write $F'(t, \sigma)$ or
$\frac{\partial}{\partial t} F(t, \sigma)$ for the derivative in the $\bS^1$-direction.

Our first main result is that we classify the central extensions of the Lie algebra $LC^{\infty}_0(\bS^2)$, 
and determine the lattice of integrable classes.
\begin{TheoremstarA}\label{TheoremA} $H^2(LC^{\infty}_0(\bS^2), \R)$ is 1-dimensional.
For every continuous 2-cocycle $\psi$ on $LC^{\infty}_0(\bS^2)$, there exists a real number $c_{\infty}\in \R$ such that $\psi$ is cohomologous to 
$c_{\infty} \psi_{\infty}$, where $\psi_{\infty}$ is given by
\begin{equation}\label{eq:cocycleTheoremA}
\psi_{\infty}(F,G) :=  \frac{12\pi}{\vol_{\mu}(\bS^2)^3}\int_{\bS^1}\Big(\int_{\bS^2} F \partial_{t}G \mu\Big)  dt.
\end{equation}
The Lie algebra extension that corresponds to this cocycle integrates to a central $\U(1)$-extension of the simply connected cover of the connected Lie group
$L\SDiff(\bS^2)_0$ if and only if $c_{\infty}$ is an integer.
\end{TheoremstarA}

The proof is carried out in Section~\ref{sec:loopSec2}, combining results from \cite{NW08} and \cite{JV16} with an essential new ingredient: a classification of invariant bilinear forms on the Poisson Lie algebra $C^{\infty}(\bS^2)_0$. In Section~\ref{sec:3twisted} we prove a similar result 
for the twisted loop group $L_{\Phi}\SDiff(\bS^2, \mu)$, 
the group of smooth functions $g \colon \R \rightarrow \SDiff(\bS^2, \mu)$ that satisfy 
$g_{t+1} = \Phi \circ g_{t}\circ \Phi^{-1}$ for an orientation-reversing `twist' $\Phi \colon \bS^2 \rightarrow \bS^2$. 

In the second part of the paper (sections \ref{Section:Recap} and \ref{Section:limits}), 
we show that the cocycles $\psi_{\infty}$ on $LC^{\infty}(\bS^2)_0$ are `fuzzy sphere limits' of cocycles $\psi_k$
on $L\su(k+1)$ under rescaling by $6/k^3$. 
For this we make extensive use of the K\"ahler structure on $\bS^2$, so in this part of the paper we identify $\bS^2$
with $\CP^1$ and take $\mu$ to be the curvature $\omega$ of the dual of the canonical line bundle.

The loop algebra $L\su(2)$ is a Lie subalgeba of $LC^{\infty}_0(\bS^2)$ in a natural fashion, 
since the inclusion $\iota \colon \su(2) \hookrightarrow C^{\infty}_{0}(\CP^1)$ loops to an inclusion 
$L\iota \colon L\su(2)\hookrightarrow LC^{\infty}_0(\bS^2)$. 
It is also a Lie subalgebra of $L\su(k+1)$.
The spin $s=k/2$ representation $(\pi_k, V_k)$ of $\su(2)$ can be considered as Lie algebra homomorphism 
$\pi_k \colon \su(2) \rightarrow \su(k+1)$, which loops to a Lie algebra homomorphism 
$L\pi_k \colon L\su(2) \rightarrow L\su(k+1)$.

Although $L\su(k+1)$ is not a Lie subalgebra of $LC^{\infty}_0(\bS^2)$, it relates to 
$LC^{\infty}_0(\bS^2)$ by the \emph{geometric quantization} map
\[Q_k \colon C^{\infty}(\bS^2) \rightarrow \mathrm{End}(V_k),\]
which we recall in Section~\ref{Section:Recap}.
More precisely, the rescaled and centered quantization map \[d\overline{\Phi}_k(f) := k\Big(Q_k(f)- {\textstyle \frac{1}{k+1}}\tr(Q_k(f)\one\Big)\]
fits in a commutative diagram
\begin{equation}\label{CDLieAlgebrawithouttraceIntro}
\begin{tikzcd}[column sep=tiny]
LC_0^{\infty}(\CP^1) \arrow[rr, dashed, "L\overline{d\Phi}_k"] &  & L\su(V_k)\\
& L\su(2),\arrow[ul, hook, "L\iota"]\arrow[ur, hook', "L\pi_k"']
\end{tikzcd}
\end{equation}
where  $L\iota$ and $L\pi_k$ are Lie algebra homomorphisms, but the horizontal map 
$L\overline{d\Phi}_k$ is not.
Nonetheless, we prove in Section~\ref{Section:limits} that the cocycles 
\begin{eqnarray}
\psi_k(X,Y) &=& \frac{1}{2\pi}\int_{\bS^1}\kappa_{\su(k+1)}(X, \partial_tY) dt\\
\psi_{\infty}(F,G) &=& \frac{6}{(2\pi)^2}\int_{\bS^1}\Big(\int_{\bS^2}F\partial_tG\omega\Big)dt
\end{eqnarray}
on 
$L\su(k+1)$ and $LC_0^{\infty}(\bS^2)$ can be arranged in a diagram 
\begin{equation}\label{CDcocycleswithouttrace}
\begin{tikzcd}[column sep=tiny]
c_{\infty} \psi_{\infty} \arrow[dr, two heads, "L\iota^*"']&  & \arrow[dl, two heads, "L\pi^*_k"] \arrow[ll, dashed, "L\overline{d\Phi}^*_k"'] c_k \psi_{k}\\
& c_1 \psi_{1}
\end{tikzcd}
\end{equation}
with the following compatibility conditions.
\begin{TheoremstarB}\label{TheoremstarB}
The Lie algebra cocycles $\psi_1$, $\psi_k$ and $\psi_{\infty}$ integrate to the group level if and only if the 
corresponding constants $c_1$, $c_k$ and $c_{\infty}$ are integers. 
\begin{itemize}
\item[a)]
We have $(L\iota)^* c_{\infty}\psi_{\infty} = c_1 \psi_1$ if and only if 
\[c_{1} = - c_{\infty}.\]
\item[b)]
We have $(L\pi_k)^* c_k \psi_k = c_1 \psi_1$ if and only if 
\[\textstyle c_{1} = \frac{k(k+1)(k+2)}{6} c_k.\]
\item[c)] 
Finally, if $c_k = - \frac{6}{k(k+1)(k+2)}c_\infty$ we have the (pointwise) limit
\[\lim_{k \rightarrow \infty} (Ld\overline{\Phi}_k)^* c_k \psi_k(f,g) = c_{\infty} \psi_{\infty}(f,g).\]
\end{itemize}
\end{TheoremstarB}
Essentially, then, the centrally extended Lie algebra $\widehat{LC^{\infty}_0(\CP^1)}$ with central charge $c_{\infty} \in \Z$
can be approximated by the affine Kac-Moody algebra $\widehat{L\su}(k+1)$ with central charge 
\begin{equation}
c_k = \frac{-c_{\infty}}{ {\textstyle \frac{1}{6}}k(k+1)(k+2)}
\end{equation}
in the limit of $k\rightarrow \infty$. 
In Section~\ref{sec:twistlim} we show that this limiting procedure is equivariant 
with respect to orientation-reserving diffeomorphisms with $\Phi^*\omega = -\omega$, and as a result we obtain 
these `fuzzy sphere limits' also in the twisted case. 

Both for the classification result and for the `fuzzy sphere limits', we believe it is possible to replace 
$\bS^2$ by more general symplectic manifolds $(\Sigma, \omega)$. In Section~\ref{sec:BeyondSpheres}
we briefly comment on the new features one might expect in this setting.

\subsection{Towards QFTs in 2+1 dimensions}

Toward the program of constructing mathematically rigorous, nontrivial QFTs in dimension $2+1$ in the setting of AQFT, this suggests the following approach.

A sequence of highest wight representations $(V_{\lambda_{k}}, \rho_{\lambda_k}, \Omega_k)$
for the affine Kac-Moody algebras $\widehat{L\su}(k+1)$ gives rise to a sequence of states $\phi_k$ on their respective universal enveloping algebras. 
Using geometric quantization, it should be possible to convert this to a sequence of states on
 the universal enveloping algebra of $\widehat{LC^{\infty}_{0}(\bS^2)}$. 
 Our results on `fuzzy sphere limits' raise the question if it is possible, using 
 a judicious choice of representations and appropriate scaling by factors of $k^{-3/2}$, to obtain convergence 
 to a state $\phi_{\infty}$ on the universal enveloping algebra of $\widehat{LC^{\infty}_0(\bS^2)}$.
 If so, the GNS construction would allow one to reconstruct a limiting Hilbert space $\cH_{\infty}$, and a limiting projective Lie algebra representation of $LC_0^{\infty}(\bS^2)$. 
 
If the resulting representation integrates to a representation $\rho_{\infty}$ at the group level, one can assign to an open subset $\mathcal{O} \subseteq \bS^1 \times \bS^2$ the algebra $\mathcal{A}(\mathcal{O})$
generated by $\rho_{\infty}(g)$ with group elements $g$ that are supported in $\mathcal{O}$. Since the Lie algebra cocycle is \emph{local}, disjoint open sets $\mathcal{O} \cap \mathcal{O}' = \emptyset$
give rise to commuting operator algebras.
One can investigate to which extent these nets of operator algebras satisfy the other Haag-Kastler axioms. Equivariance with respect to the isometry group $\mathrm{SO}(2)\times \mathrm{SO}(3)$ would likely follow from the representation theory, but whether or not the net of operators will be equivariant for $\mathrm{SO}(3,2)$ would probably be harder to answer. 
\section{Loop groups with values in $\SDiff(\bS^2,\mu)$}\label{sec:loopSec2}

The group $K =\SDiff(\bS^2,\mu)$ of diffeomorphisms that preserve the $\mathrm{SO}(3)$-invariant volume form  $\mu$ on $\bS^2$
is a Fr\'echet--Lie group. Its Lie algebra $\fk = \X(\bS^2,\mu)$ is the Lie algebra of divergence-free vector fields, which we identify 
with the Poisson algebra $C_0^{\infty}(\bS^2)$ of smooth functions on $\bS^2$ that integrate to zero against $\mu$. The Poisson bracket 
comes from $\mu$, considered as a symplectic form on $\bS^2$.

The \emph{loop group} $LK = L\SDiff(\bS^2, \mu)$ 
is the group of maps $\phi \colon \bS^1 \rightarrow \SDiff(\bS^2,\mu)$ 
that are smooth with respect to the Fr\'echet--Lie group structure on $\SDiff(\bS^2,\mu)$. Equivalently, $L\SDiff(\bS^2, \mu)$ is
is the group of all vertical automorphisms that respect the orientation and volume form on every fibre of the trivial bundle $\bS^1 \times \bS^2 \rightarrow \bS^1$.
 
 It  is a Fr\'echet--Lie group with Lie algebra 
$L\fk = C^{\infty}(\bS^1,\fk)$. Since $\fk = C^{\infty}_{0}(\bS^2)$, 
we can identify 
\[
L\fk = C_0^{\infty}(\bS^1\times \bS^2) \subseteq C^{\infty}(\bS^1 \times \bS^2)
\] 
with the subalgebra of the Poisson algebra $C^{\infty}(\bS^1 \times \bS^2)$ consisting of functions that integrate to zero 
on every fibre $\bS^2_t := \{t\} \times \bS^2 \subseteq \bS^1 \times \bS^2$, equipped 
with the (degenerate)
bracket depending only on the $\bS^2$-variables:
\[
	\{F,G\}(t,\vec{x}) := \{F_t, G_t\}(\vec{x}).
\]
The bracket on the left hand side is a Poisson bracket 
on the ring $C^{\infty}(\bS^1\times\bS^2)$ (with the spheres as symplectic leaves), the
 bracket on the right hand side is the 
ordinary Poisson bracket on $\bS^2$ (from the symplectic form $\mu$) between the functions 
$F_{t} \colon \bS^2 \rightarrow \R$ and $G_t \colon \bS^2 \rightarrow \R$ at fixed $t\in \bS^1$.
 
 \subsection{Central extensions and 2-Cocycles}

For a locally convex Lie algebra $\fg$, we denote by
$H^{\bullet}(\fg, \R)$ the \emph{continuous} Lie algebra cohomology with trivial coefficients, that is, 
the cohomology of the complex $C^{\bullet}(\fg, \R)$ of alternating multilinear maps 
$\fg \times \ldots \times \fg \rightarrow \R$
that are jointly continuous with respect to 
the Fr\'echet topology.  The differential is the ordinary Chevalley-Eilenberg differential, 
of which the continuous cochains are a subcomplex because the Lie bracket $\fg \times \fg \rightarrow \fg$ is continuous.
The locally convex central extensions of $\fg$ by $\R$ are classified by $H^2(\fg, \R)$.

On $\fk = C^{\infty}_{0}(\bS^2)$ we define the inner product 
\[
	\kappa(f,g) := \int_{\bS^2}fg \,\mu.
\]
The inner product is invariant under $\SDiff(\bS^2)$ because $\mu$ is invariant, 
\[\textstyle
	\kappa(\phi^*f, \phi^*g) = \int_{\bS^2}(\phi^*f \phi^*g) \mu = \int_{\bS^2}\phi^*(fg\mu) = \kappa(f,g)
\]
for all $\phi \in \SDiff(\bS^2)$. In particular it is invariant under the adjoint Lie algebra action, 
\[\kappa(\{h, f\}, g) + \kappa(f, \{h,g\}) = 0.\]

We prove the following result, which shows that the inner product $\kappa$ takes the role of the Killing form 
in the second second continuous cohomology $H^2(L\fk, \R)$ for $\fk = C^{\infty}_0(\bS^2)$.

\begin{Theorem}\label{ThmCocycleLoop} $H^2(L\fk, \R)$ is 1-dimensional.
Every continuous 2-cocycle $\psi$ on $L\fk$ is cohomologous to $c\psi_{\infty}$,
where $c$ is a real number and $\psi_{\infty}$ is the cocycle
\begin{eqnarray}\label{eq:cocycleTheorem1}
\psi_{\infty}(F,G) &:=&  \frac{12\pi}{\vol_{\mu}(\bS^2)^3}\int_{\bS^1}\kappa \left( F, {\textstyle \frac{\partial}{\partial t} }G\right)  dt\\
& = & \frac{12\pi}{\vol_{\mu}(\bS^2)^3}\iiint_{\bS^1 \times \bS^2} F(t, \sigma) {\textstyle \frac{\partial}{\partial t} }G(t, \sigma) \mu(d\sigma)dt. \nonumber
\end{eqnarray}
\end{Theorem}

By \cite{NW08}, the continuous 2-cocycles on $L\fk$ for an arbitrary Lie algebra $\fk$ can be understood in terms of 
the continuous second cohomology $H^2(\fk, \R)$, the space
$\mathrm{Sym}^2(\fk, \R)^{\fk}$ of continuous invariant bilinear forms on $\fk$, and the 
Koszul map $\Gamma \colon \mathrm{Sym}^2(\fk, \R)^{\fk} \rightarrow H^3(\fk, \R)$
that sends $S$ to $(\Gamma S)(f,g,h) = \kappa(\{f, g\}, h\})$. 
We prove Theorem~\ref{ThmCocycleLoop} by describing each of these ingredients, and then combining the information using 
the results in \cite{NW08}.

The continuous second Lie algebra cohomology of the the Poisson algebra $C^{\infty}_0(\Sigma)$
for a general symplectic manifold $(\Sigma, \omega)$ has been dealt with in \cite{JV16}.

\begin{Theorem}[\cite{JV16}]\label{Thm:H2Poisson}
If $\Sigma$ is a compact symplectic manifold, then \begin{equation}H^2(C^{\infty}_0(\Sigma), \R) = H^1_{\mathrm{dR}}(\Sigma).\end{equation}
\end{Theorem} 
In particular, this implies that $H^2(\fk, \R) = 0$ for $\fk = C^{\infty}_{0}(\bS^2)$. 
Although many of the arguments in this paper carry over to general symplectic manifolds, 
we will mainly stick to the case $\Sigma = \bS^2$ from now on. 
In Section~\ref{sec:BeyondSpheres}, we will comment briefly on the possible new features one can expect for
symplectic manifolds where  $H^1_{\mathrm{dR}}(\Sigma)$ is nontrivial.
 
 \subsection{Invariant symmetric bilinear forms on $C^{\infty}_0(\bS^2)$}
 
We show that up to scaling, the invariant symmetric bilinear form
\begin{equation}\label{eq:TheInnerProduct}
\kappa(f,g) = \int_{\bS^2}fg\;\mu
\end{equation}
is the unique continuous invariant bilinear form on $C^{\infty}_0(\bS^2)$.
 In fact, rather than $\fk = C^{\infty}_0(\bS^2)$, it will be convenient to consider the full Poisson algebra
$C^{\infty}(\bS^2)$. We first show that this does not essentially change the situation.

\subsubsection{The role of the centre}
Let $(\Sigma,\omega)$ be a compact symplectic manifold, and let $S \colon C^{\infty}(\Sigma) \times C^{\infty}(\Sigma) \rightarrow \R$ be
a continuous invariant symmetric bilinear form on the Poisson algebra $C^{\infty}(\Sigma)$.
Since $\Sigma$ is compact, $C^{\infty}(\Sigma) = \R \one \oplus C^{\infty}_0(\Sigma)$ 
is a direct sum of Lie algebras.

Every invariant symmetric bilinear form on $C^{\infty}(\Sigma)$ restricts to an 
invariant symmetric bilinear form on $C^{\infty}_{0}(\Sigma)$. 
Conversely, since constant functions are central, 
every invariant symmetric bilinear form $S$ on $C^{\infty}_{0}(\Sigma)$ extends to an invariant bilinear form 
on $C^{\infty}(\Sigma)$ by putting $S(\one, f) = S(f,\one) =0$ for all $f\in C^{\infty}(\Sigma)$. 

This is not a bijective correspondence: for instance, the invariant symmetric bilinear form 
\begin{equation}\label{eq:FormOnCentre}
B(f, g) := \Big(\int_{M}f\omega^n/n!\Big)\Big(\int_{M}g\omega^n/n!\Big)
\end{equation}
clearly restricts to zero on $C^{\infty}_{0}(\Sigma)$. However, according to the following proposition this is 
essentially the only counterexample.

\begin{Proposition}
If $M$ is compact, then
two invariant bilinear forms $S$ and $S'$ on $C^{\infty}(\Sigma)$ restrict to the same invariant bilinear 
form on $C_0^{\infty}(\Sigma)$ if and only if they differ by a multiple of $B$.
\end{Proposition}
\begin{proof}
Suppose that $S - S'$ restricts to zero on $C^{\infty}_0(\Sigma)$. 
Then $S'' := S - S' - c_{B}B$ additionally satisfies $S''(\one, \one) = 0$ if we set
$c_{B} := (S(\one,\one) - S'(\one,\one))/\vol^2(\Sigma)$.
Since $S''$ is invariant and $1$ is central, $S''(1, \{q,r\}) = S''(\{1,q\}, r) = 0$ for all $q, r \in C^{\infty}(\Sigma)$.
But by \cite[\S 12]{ALDM74} the commutator ideal of $C^{\infty}(\Sigma)$ is $C^\infty_{0}(\Sigma)$, 
so $S''(\one, f) = 0$ for $f \in C^{\infty}_{0}(\Sigma)$.  
But since $S''(\one,\one) = 0$, we have $S''(1, f) = 0$ for all $f \in C^{\infty}(\Sigma)$, so $S'' = 0$
because it vanishes on $C^{\infty}_{0}(\Sigma)$.
\end{proof}

Classifying invariant bilinear forms on $C^{\infty}_0(\Sigma)$ is therefore equivalent to classifying invariant 
bilinear forms on $C^{\infty}(\Sigma)$, and in the following we will take the second point of view.

\subsubsection{The Poisson algebra of the 2-sphere.}

We now consider the sphere $\bS^2$, equipped with a nonzero 
$\mathrm{SO}(3)$-invariant symplectic form $\mu$.
The Poisson algebra $C^{\infty}(\bS^2)$ has a dense polynomial subalgebra $\R[\bS^2] := \R[x,y,z]/I_{\R}$, where $I_{\R}$ 
is the ideal generated by $x^2 + y^2 + z^2 - 1$. 
Every continuous invariant bilinear form on $C^{\infty}(\bS^2)$ is uniquely determined by its
restriction to $\R[\bS^2]$. These correspond to 
$\C$-linear invariant bilinear forms on the complexification $\C[\bS^2] = \C[x, y, z]/I_{\C}$.
 
 Identifying $\C^3$ with the linear Poisson manifold $\mathfrak{so}(3,\C)^*$, the Poisson bracket on $\C[x,y,z]$
 is given by
 \begin{equation}\label{eq:KKSbracket}
 \{p,q\}_{\vec{r}} = \left\langle
 \vec{r}\;,\; 
 \vec{\nabla}p \times \vec{\nabla}q
 \right\rangle.
 \end{equation}
 Here $I_{\C}$ is a Poisson ideal, and $\C[x,y,z]/I_{\C}$ coincides with the Poisson algebra for the 
 Kostant-Kirillov-Souriau symplectic structure $\mu_{\mathrm{KKS}}$ on 
 $\bS^2$.
 
\begin{Theorem}\label{Thm:InvBilS2Alg}
Every invariant bilinear form on $\C[\bS^2]$ is of the form
\[S = c_{\kappa} \kappa + c_{B} B\] 
for some $c_{\kappa}, c_{B} \in \C$.
\end{Theorem} 

Our proof relies on the $\mathfrak{so}(3,\C)$-equivariance of the Poisson structure. 
As an $\mathfrak{so}(3,\C)$-representation, $\C[\bS^2]$ decomposes into irreducible representations as
\[
\C[\bS^2] = \bigoplus_{l=0}^{\infty} V_{l},
\]
where the spin $l$-representation $V_{l} \subset \C[\bS^2]$. 
The quotient map $\C[x,y,z] \rightarrow \C[\bS^2]$ induces an isomorphism $\cH_{l} \stackrel{\sim}{\longrightarrow} V_l$
with the space $\mathcal{H}_l \subset \C[x,y,z]$ of harmonic polynomials in $x,y,z$ that are homogeneous of degree $l\in \N$, see e.g.\
\cite[\S IV.9, Problem~12]{Kn96} and \cite[\S IV.9, Problem 2]{Kn96}.
Since $V_1 \subset \C[\bS^2]$ is isomorphic to $\mathfrak{so}(3,\C)$, the Poisson bracket 
$\{\,\cdot\, , \,\cdot\,\} \colon \C[\bS^2] \wedge \C[\bS^2] \rightarrow \C[\bS^2]$ is an intertwiner of $\mathfrak{so}(3,\C)$-representations
by the Jacobi identity.  
 
The $\mathfrak{so}(3,\C)$-representation $P_{l} \subseteq \C[x,y,z]$ of homogeneous polynomials of degree $l$ decomposes as
\[
P_l = \cH_l \oplus r^2 \cH_{l-2} \oplus r^4 \cH_{l-4}\ldots, 
\] 
where $r^2 = x^2 + y^2 + z^2$ is of course rotation invariant. Since the KKS bracket \eqref{eq:KKSbracket} is 
of degree $-1$, we have $\{P_k, P_l\} \subseteq P_{k+l-1}$, so in particular $\{\cH_k, \cH_l\} \subseteq P_{k+l-1}$.
As the Poisson bracket is an intertwiner $\cH_k \otimes \cH_l \rightarrow P_{k+l-1}$, its image 
can only contain irreducible representations of spin \[|k-l|, \; |k-l| +1, \; \ldots,\; k+l -1,\; k+l.\] 
In particular, the Poisson bracket restricts to an intertwiner 
\begin{equation}
\{\,\cdot\,, \,\cdot\,\}_{2l} \colon \cH_2 \otimes \cH_{l} \rightarrow \cH_{l+1} \oplus \cH_{l-1}.
\end{equation}
\begin{Lemma}\label{Lemma:Weightl}
For $l > 0$,
the image of this intertwiner contains $\cH_{l+1}$.
\end{Lemma}
\begin{proof}
The vector $z^2 \in \cH_2$ of weight 0 and the vector $(x+iy)^l \in \cH_{l}$ of (highest) weight $l$
have nonzero commutator $\{z^2, (x+iy)^l\} = -2ilz(x+iy)^l$ of weight $l$. Since the image of the 
intertwiner is not contained in the representation $\cH_{l-1}$ of highest weight $l-1$, it 
contains $\cH_{l+1}$ by Schur's Lemma.
\end{proof} 
\begin{proof}[Proof of Theorem~\ref{Thm:InvBilS2Alg}]
Let $S \colon \C[\bS^2] \times \C[\bS^2] \rightarrow \C$ be an invariant bilinear form. 
Since its restriction to $S_{kl} \colon V_k \otimes V_l \rightarrow \C$ is $\mathfrak{so}(3,\C)$-invariant, 
is induces an intertwiner $V_k \rightarrow V_l^*$. 
Then by Schur's Lemma $S_{kl} = 0$ for $k\neq l$, 
and $S_{kk}$ is a scalar multiple of $\kappa_{kk}$, where 
$\kappa \colon \C[\bS^2] \times \C[\bS^2] \rightarrow \C$ is the (complexification of the) 
nondegenerate symmetric bilinear form defined in~\eqref{eq:TheInnerProduct}.
In particular, $S$ is symmetric.

The bilinear forms $\kappa$ and $B$ from equation \eqref{eq:TheInnerProduct} and \eqref{eq:FormOnCentre}
satisfy $B_{00} \neq 0$, $B_{11} = 0$, $\kappa_{00} \neq 0$ and $\kappa_{11} \neq 0$, 
so by subtracting suitable multiples of $\kappa$ and $B$ we may assume that $S_{00} = 0$ and $S_{11} = 0$.
Assume by induction that $S_{ii} = 0$ for all $i\leq l$. Then $S(V_{l}, \{V_2, V_{l+1}\}) = 0$
because $S_{lk} = 0$ for all $k$. The invariance then implies that 
$S(\{V_2, V_l\}, V_{l+1}) = 0$ as well. For $l>0$ we have $V_{l+1} \subseteq \{V_2, V_l\}$ by Lemma~\ref{Lemma:Weightl}, 
so $S_{l+1, l+1} = 0$ as required.
\end{proof}
 
We will continue in the smooth setting, for which we need the following corollary. It follows from Theorem~\ref{Thm:InvBilS2Alg} 
and the fact that $\R[\bS^2]$ is dense Lie subalgebra of $C^{\infty}(\bS^2)$.
\begin{Corollary}\label{Thm:InvBilS2}
Every continuous invariant bilinear form on $C^{\infty}(\bS^2)$ is of the form
\[S = c_{\kappa} \kappa + c_{B} B\] 
for some $c_{\kappa}, c_{B} \in \R$. Every continuous invariant bilinear form 
on $C^{\infty}_0(\bS^2)$ is of the form $S = c_{\kappa}\kappa$.
\end{Corollary} 

\begin{Remark}\label{RemarkGeneralSigma}
Although we extensively used the $\mathfrak{so}(3)$-invariance in the proof, we believe that 
the result should hold true for the compactly supported Poisson algebra $C^{\infty}_{c}(\Sigma)$ 
of an arbitrary symplectic manifold, without any additional symmetry. 
Using techniques similar to \cite[Prop.~4.2]{JV16} one can prove that every bilinear 
form is \emph{local}, so a local model (with additional symmetry) should suffice to prove the general result.
\end{Remark}
 
\subsection{The Koszul map}

Note that the restriction of $\kappa \colon C^{\infty}(\bS^2) \times C^{\infty}(\bS^2) \rightarrow \R$
to the subalgebra $V_1 \simeq \mathfrak{so}(3)$ is an invariant symmetric bilinear form on the 
simple Lie algebra $\mathfrak{so}(3)$, and hence equal to a multiple of the killing form.
Since $\kappa$ is nondegenerate, this multiple is nonzero.

 The Koszul map $\Gamma \colon \mathrm{Sym}^2(\fk, \R)^{\fk} \rightarrow H^3(\fk, \R)$ sends 
 a continuous invariant symmetric bilinear form $S$ to the class of the continuous 3-cocycle 
 $(\Gamma S) (f, g, h) = S(\{f, g\}, h)$.
 \begin{Proposition}\label{Prop:Koszul}
 For $\fk = C^{\infty}_0(\bS^2, \R)$ the Koszul map is injective on the continuous bilinear forms.
 For its dense subalgebra $\fk = \R[\bS^2]_0$ of polynomials that integrate to zero, the Koszul map is 
 injective on all bilinear forms.
 \end{Proposition}
 \begin{proof}
 Identify the homogeneous polynomials of degree one on $\bS^2$ with $\mathfrak{so}(3)$, and denote 
 the corresponding inclusion by $\iota \colon \mathfrak{so}(3) \rightarrow \fk$.
 Denote the Koszul maps for $\mathfrak{so}(3)$ and $\fk$ by $\Gamma_{\mathfrak{so}(3)}$
 and $\Gamma_{\fk}$, respectively.  
 Since the Koszul map is functorial, $\iota^* \circ \Gamma_{\fk} = \Gamma_{\mathfrak{so}(3)} \circ \iota^*$, 
 and $\iota^* (\Gamma_{\fk} \kappa)$ is a nonzero multiple of  $\Gamma_{\mathfrak{so}(3)}(\kappa_{\mathfrak{so}(3)})$.
 
 This class generates $H^3(\mathfrak{so}(3), \R)$, so in particular it is nonzero. 
 This implies that $\Gamma_{\fk}(c\kappa) = 0$ if and only if $c=0$. Since every (continuous) element of 
 $\mathrm{Sym}^2(\fk, \R)^{\fk}$ is of this form, the result follows.
 \end{proof}
 
 \subsection{2-cocycles on $C^{\infty}_{c}(M, \fk)$ and the proof of Theorem~\ref{ThmCocycleLoop}}\label{sec:proofThm1}

In principle, the results from \cite{NW08} provide a classification of 2-cocycles on general current algebras 
$A \otimes \fk$, where $A$ is any unital commutative algebra and $\fk$ is any Lie algebra. 
Their main result states that cocycles on $A\otimes \fk$ essentially come from two sources: 
invariant bilinear forms on $\fk$, and 2-cocycles on $\fk$.  The way in which these two sources give rise to cocycles on $A\otimes \fk$  
is then governed by the 
Koszul map. 

For the Poisson Lie algebra $\fk = C^{\infty}_0(\bS^2)$ of smooth functions that average to zero, 
we now essentially control these two sources:
we have classified the invariant bilinear 
forms in Theorem~\ref{Thm:InvBilS2Alg}, and the continuous 2-cocycles were already classified in \cite{JV16} 
(Theorem~\ref{Thm:H2Poisson}).
Together with the injectivity of the Koszul map (Proposition~\ref{Prop:Koszul}), the results in  \cite{NW08}
should therefore allow us to complete the classification of 2-cocycles for $C^{\infty}(\bS^1, \R)\otimes \fk$, 
which is a dense subalgebra of $L\fk = C^{\infty}(\bS^1, \fk)$.

A minor complication in completing the proof of Theorem~\ref{sec:proofThm1} 
is that the setting in \cite{JV16} (continuous cochains) does not quite match
with the setting in \cite{NW08} (cochains without continuity assumptions).
For this reason we will have to carefully navigate the various types of continuity properties that we use in different parts of the proof.

In the following we denote by $\Omega^1_A$ the K\"ahler differentials of $A$, with universal derivation $d \colon A \rightarrow \Omega^1_{A}$.
We say that a $q$-cochain $\psi$ on $A\otimes \fk$ is \emph{continuous in the $\fk$-variables} if 
the multilinear maps $(f_1, \ldots, f_q) \mapsto \psi(a_1\otimes f_1, \ldots, a_q \otimes f_q)$ are continuous for fixed 
$a_1, \ldots, a_q$.
\begin{Lemma}\label{Lem:NeebWagemannRefined}
Let $\fk = C^{\infty}_0(\bS^2)$, and let $\psi$ be a 2-cocycle on $A\otimes \fk$ that is continuous in the $\fk$-variables.
 Then $\psi$ is cohomologous to 
\begin{equation}\label{eq:CocycleForGeneralRing}
\psi_{\lambda}(a\otimes f, b \otimes g) = \lambda(adb)\kappa(f,g)
\end{equation}
for a linear map $\lambda \colon \Omega^1_{A}/dA \rightarrow \R$.
\end{Lemma}
\begin{proof}
A pair of linear maps
\begin{align*}
 f_1 \colon& \Lambda^2(A) \otimes S^2(\fk) \rightarrow \R,\\
 f_2 \colon & A \otimes \Lambda^2(\fk) \rightarrow \R
 \end{align*}
gives rise to a pair of bilinear maps $\psi_1, \psi_2$ on $A\otimes \fk$ by
\begin{align*}
 \psi_1(a\otimes f, b\otimes g) & = f_1(a\wedge b \otimes f \vee g)\\
 \psi_2 (a\otimes f, b\otimes g) & =  f_2(ab \otimes f\wedge g),
 \end{align*}
where $\wedge$ and $\vee$ denote the alternating and symmetric tensor product.
Since the commutator subalgebra $\fk' := [\fk,\fk]$ is equal to $\fk$
by \cite{ALDM74}, it follows from \cite[Thm.~4.1(iv)]{NW08} that every 2-cocycle on $A\otimes \fk$
is a sum $\psi = \psi_1 + \psi_2$ of bilinear maps of this form.
If $\psi$ is continuous in the $\fk$-variables, 
then the bilinear map $(f,g) \mapsto f_1(a\wedge b \otimes f \vee g)$ is continuous as well, and it is 
$\fk$-invariant by \cite[Thm.~4.1(i)]{NW08}. 
So by Corollary~\ref{Thm:InvBilS2} it is a multiple of $\kappa$, and we have 
$f_1(a\wedge b\otimes f\vee g) = L(a\wedge b)\kappa(f,g)$ 
for a linear map $L \colon A \wedge A \rightarrow \R$. By 
\cite[Thm.~4.1(ii) and Lemma~2.1]{NW08} we have $L(a\wedge b) = \lambda(adb - bda)$ for a linear map 
$\lambda \colon \Omega^1_{A} \rightarrow \R$.

Since the Koszul map for $\fk$ is injective on continuous cocycles (Proposition~\ref{Prop:Koszul}), it follows from
\cite[Thm.~4.7(iii)]{NW08} that $f \wedge g \mapsto f_2(a\otimes f \wedge g)$ is a 2-cocycle on $\fk$ for every $a\in A$. 
If $\psi$ is continuous in the $\fk$-variables, then this map is a continuous 2-cocycle on $\fk$, and therefore a coboundary
by Theorem~\ref{Thm:H2Poisson}. So $f_2(a\otimes f \wedge g) = \gamma_{a}([f,g])$ for a linear map 
$\gamma \colon A \rightarrow \fk'$ into the continuous linear dual of $\fk$. 
In particular $\psi_2(a\otimes f, b\otimes g) = \gamma([a\otimes f, b\otimes g])$ is the coboundary
of the 1-cocycle $\gamma(a\otimes f) := \gamma_{a}(f)$, and $\psi$ is cohomologous to 
$\psi_1(a\otimes f, b\otimes g) = \lambda(adb - bda)\kappa(f,g)$. 
By \cite[Thm.4.7(ii)]{NW08} the map $(a,b) \mapsto \lambda(adb - bda)$ is a cyclic 1-cocycle,
which implies that $\lambda(dA) = 0$.
\end{proof}
\begin{Remark}\label{Rk:partCont1chain}
In fact, we have proven the slightly stronger statement that $\psi - \psi_{\lambda} = d \gamma$ for a $1$-cochain $\gamma \colon A\otimes \fk \rightarrow \R$
that is continuous in the $\fk$-variables. 
\end{Remark}

Applied to the ring $A = \C[z, z^{-1}]$ of Laurent polynomials, this immediately yields a proof of Theorem~\ref{ThmCocycleLoop}.
\begin{proof}[Proof of Theorem~\ref{ThmCocycleLoop}]
We identify the Laurent monomial $z^n \otimes f$ in $\C[z,z^{-1}] \otimes_{\R} \fk$ with the Fourier mode $t \mapsto e^{int}f$ in $L\fk_{\C}$, 
and denote the corresponding inclusion map
by $\iota \colon \C[z,z^{-1}] \otimes \fk \rightarrow L\fk_{\C}$. If $\psi$ is a continuous Lie algebra cocycle on $L\fk$, then 
its complexification $\psi_{\C}$ pulls back to a $\C$-linear cocycle $\iota^*\psi_{\C}$ on ${\C[z,z^{-1}] \otimes_{\R} \fk}$ that is 
continous in the $\fk$-variables.

For $A = \C[z, z^{-1}]$, the space $\Omega^1_{A}/dA$ is 1-dimensional with generator $[z^{-1}dz]$, so every linear map 
$\lambda \colon \Omega^1_{A}/dA \rightarrow \C$ is proportional to $\lambda(adz) = \frac{1}{2\pi i}\oint adz$.
By Lemma~\ref{Lem:NeebWagemannRefined}, every pullback $\iota^*\psi_{\C}$ of a continuous cocycle 
is cohomologous to $c\psi_{\lambda}$ for some $c\in \C$. Evaluating this on the real subalgebra $(\C[z,z^{-1}] \otimes_{\R} \fk) \cap L\fk$ (on which $\psi$
takes real values),  we find 
that $c\in \R$.

The complexification of the continuous cocycle $\psi_{\infty}$ on $L\fk$ from \eqref{eq:cocycleTheoremA} 
pulls back to a cocycle that is proportional to 
\[\psi_{\lambda}(a\otimes f, b \otimes g) = \Big(\frac{1}{2\pi i}\oint adb\Big) \;\kappa(f,g),\]
so there exists a constant $c$ such that $\iota^* (\psi_{\C} - c\psi_{\infty})$ is exact, the boundary of a 1-cochain $\gamma \colon \C[z,z^{-1}] \otimes \fk \rightarrow \C$
that is continuous in the $\fk$-variables.

To show that $\gamma$ extends to a continuous cochain, note that 
\begin{equation}\label{eq:get1from2}
\gamma(a \otimes \{f,g\}) = {\textstyle\frac{1}{2}}(\psi_{\C}(a\otimes f, \one\otimes g) - \psi_{\C}(a\otimes g, \one \otimes f)).
\end{equation}
Since $\fk$ is perfect and $\psi$ is continuous, the map $a \mapsto \gamma(a\otimes f)$ on $\C[z^{-1}, z]$ extends
to a continuous map 
$C^{\infty}(\bS^1, \C) \rightarrow \C$ for every fixed $f\in \fk$.
Since $f \mapsto \gamma(a\otimes f)$ is continuous by Remark~\ref{Rk:partCont1chain},
the bilinear map 
$(a, f) \mapsto \gamma(a\otimes f)$ is separately continuous. 
Since $C^{\infty}(\bS^1, \C)$ and $\fk$ are Fr\'echet spaces, separate continuity implies joint continuity, 
and $\gamma \colon C^{\infty}(\bS^1, \C) \otimes \fk \rightarrow \C$ extends to a continuous linear functional 
on $C^{\infty}(\bS^1, \C)\otimes_{\pi}\fk = L\fk_{\C}$  by the universal property of projective tensor product.
\end{proof}

\begin{Remark}
It is natural to conjecture that $H^2(\C[\bS^2]_0, \R) = \{0\}$ for the Poisson algebra $\C[\bS^2]_0$, and conceivably one could 
prove this directly along the lines of Theorem~\ref{Thm:InvBilS2Alg}, using the invariant bilinear form 
to convert cocycles in skew-symmetric derivations. From \cite{NW08} one would then immediately get a 
classification of the (not necessarily continuous) 2-cocycles on $\C[z^{-1}, z] \otimes \C[\bS^2]_0$, 
which in turn would imply Theorem~\ref{ThmCocycleLoop} because $\C[z^{-1}, z] \otimes \C[\bS^2]_0$ is a dense subalgebra of $L\fk_{\C}$.
\end{Remark}

 \subsection{Integral classes}
 
 Having classified the continuous 2-cocycles on $L\fk$, it remains to 
 determine which cocycles give rise to Lie algebra extensions that integrate to the group level. 
 For this, it is convenient to normalise the bilinear form on $C^{\infty}(\bS^2)$ as 
 \[
 \kappa_{\infty}(f,g) := \frac{12\pi}{\vol^3(\bS^2)}\int_{\bS^2}fg\mu.
 \]
The isomorphism between the Poisson algebra $C^{\infty}_0(\bS^2)$ and the Lie algebra $\X(\bS^2, \mu)$ of divergence free vector 
fields depend on the scaling of $\mu$. If we scale $\kappa_{\infty}$ as indicated above, then the induced bilinear form $\kappa_{\infty}$ on 
$\X(\bS^2, \mu)$ will be independent of this isomorphism.
 
 \subsubsection{Relation between integrality conditions for $L\su(2)$ and $L\fk$}
 By Smale's Theorem \cite{Smale1959}, the group of orientation preserving 
 diffeomorphisms $\mathrm{Diff}(\bS^2)_0$ is homotopy equivalent to $\mathrm{SO}(3)$.
 In turn, $K = \SDiff(\bS^2,\mu)$ is homotopy 
 equivalent to $\mathrm{Diff}(\bS^2)_0$  because the orbit map
 \[\mathrm{Diff}(\bS^2)_{0} \rightarrow \Omega^2(\bS^2)_1: \phi \mapsto \phi^*\mu\] is a smooth 
 principal $\SDiff(\bS^2,\mu)$ bundle over the contractible space $\Omega^2(\bS^2)_{1}$ of volume forms that integrate to one \cite{Ham82}.

It follows that the inclusion of Lie groups $L\mathrm{SO}(3) \hookrightarrow LK$ is a homotopy equivalence.
Since $\pi_1(\mathrm{SO}(3)) = \Z_2$, both $L\mathrm{SO}(3)$ and $LK$
have two connected components.
And since $\pi_2(\mathrm{SO}(3))$ is trivial, we have $\pi_1(L\mathrm{SO}(3)) \simeq \pi_1(\mathrm{SO}(3)\times \Omega\mathrm{SO}(3)) \simeq \Z/2\Z$,
and $\Z/2\Z \rightarrow L\mathrm{SU}(2) \rightarrow L\mathrm{SO}(3)_0$ is the simply connected cover. 

Since $\mathrm{SU}(2)\hookrightarrow \widetilde{K}_0$ is a homotopy equivalence, 
the group $L\widetilde{K}_0$ is the simply connected cover of $LK_0$,
and the inclusion $L\mathrm{SU}(2) \hookrightarrow L\widetilde{K}_0$ is a homotopy equivalence
between the two 1-connected Lie groups.

Let $G$ be a connected, simply connected locally convex Lie group with Lie algebra $\fg$.
If $\mathrm{U}(1) \rightarrow G^{\sharp} \rightarrow G$ is a central extension of locally convex Lie groups, 
then it differentiates to a continuous central extension $\R \rightarrow \fg^{\sharp} \rightarrow \fg$ of locally convex Lie algebras.
Up to strict isomorphism, continuous central extensions are characterized  
by the cohomology class $[\psi] \in H^2(\fg, \R)$ of the curvature cocycle 
$\psi(x, y) := [\sigma(x), \sigma(y)] - \sigma([x,y])$, derived from a continuous linear splitting 
$\sigma \colon \fg \rightarrow \fg^{\sharp}$ (these exist by the Hahn--Banach Theorem).

Conversely, by \cite{Ne02}, a continuous Lie algebra extension $\R \rightarrow \fg^{\sharp} \rightarrow \fg$
integrates to a locally convex Lie group extension of a 1-connected Lie group $G$ if and only 
if the \emph{period homomorphism} $P_{\psi} \colon \pi_2(G) \rightarrow \R$ takes integral values.
For sufficiently smooth classes $[\sigma] \in \pi_2(G)$, this is defined by 
$P_{\psi}([\sigma]) = \int_{\bS^2}\sigma^*\Psi$, where
$\Psi$ is the unique left-invariant 2-form on $G$ that agrees with $\psi$ on $T_{\one}G \simeq \fg$.

Since $L\mathrm{SU}(2) \hookrightarrow L\widetilde{K}_0$ is an inclusion of locally convex Lie groups, 
every continuous central extension of $L\widetilde{K}_0$
restricts to a continuous central extension of $L\mathrm{SU}(2)$. 
If the former gives rise to the cocycle 
\begin{eqnarray}\label{eq:LKcocycle}
	\psi_{\infty}(F,G) &:=& \int_{\bS^1}\kappa_{\infty}(F, \partial_{t}G) dt\\
	& = & \frac{12\pi }{\vol_{\mu}(\bS^2)^3}\int_{\bS^1}\Big(\int_{\bS^2} F\partial_t G \mu\Big) dt\nonumber
\end{eqnarray}
on $L\fk$, then we will see below that this pulls back to the cocycle
\begin{equation}\label{eq:Lso3cocycle}
	\psi_{\mathfrak{su}(2)}(X,Y) = -\frac{1}{2\pi} \int_{\bS^1}\kappa_{\mathfrak{su}(2)}(X,\partial_{t}Y)dt
\end{equation}
on $L\mathfrak{so}(3) \simeq L\su(2)$, where the invariant bilinear form is scaled by $\kappa_{\mathfrak{su}(2)}(h,h) = -2$ on the coroot
$h = \mathrm{diag}(i,-i)$. 
 
 By \cite[Thm.~4.4.1(iv)]{PS86}, the cocycle $c\psi_{\mathfrak{su}(2)}$ from \eqref{eq:Lso3cocycle} corresponds
  to a central extension $L\mathrm{SU}(2)^{\sharp} \rightarrow L\mathrm{SU}(2)$ of Lie groups 
 if and only if $c\in \Z$, so this is a \emph{necessary} condition for 
 \eqref{eq:LKcocycle} to come from a Lie group extension $LK^{\sharp}_0 \rightarrow LK_0$.
 But since $\iota \colon L\mathrm{SU}(2) \hookrightarrow L\widetilde{K}_0$ induces an isomorphism 
 $\iota_* \colon \pi_2(L\mathrm{SU}(2)) \rightarrow \pi_2(L\widetilde{K}_0)$, the integrality of the period 
 homomorphism for $c\psi_{\mathfrak{su}(2)}$ from \eqref{eq:Lso3cocycle} on $L\mathrm{SU}(2)$ implies integrality of the 
 period homomorphism of $c\psi_{\infty}$ from \eqref{eq:LKcocycle} for $L\widetilde{K}_0$. So $c\in \Z$ is a \emph{sufficient} condition as well.
 So we arrive at the following result:
 
 \begin{Theorem}\label{Thm:integralityCocycles}
 The cocycle \eqref{eq:LKcocycle} corresponds to a smooth central extension of $L\widetilde{K}_0$ by $\mathrm{U}(1)$
 if and only if $c_{\infty}\in \Z$.
 \end{Theorem}
 \begin{proof}
 The only thing left to show is that the cocycle \eqref{eq:LKcocycle} restricts to \eqref{eq:Lso3cocycle} 
 on the subalgebra $L\mathfrak{so}(3)$.
 Denote by $\iota \colon \mathfrak{so}(3) \hookrightarrow \X(\bS^2, \mu)$ the inclusion of Lie algebras that comes from the  infinitesimal action of $\mathfrak{so}(3)$ on $\bS^2$. The equation $i_{X_f}\mu = -df$ gives an isomorphism between $\X(\bS^2, \mu)$ 
 and $C_0^{\infty}(\bS)$, and in this way we obtain an inclusion $\iota \colon \mathfrak{so}(3) \hookrightarrow C^{\infty}_0(\bS^2)$. 
 (The inclusion in $C^{\infty}_0(\bS^2)$ changes if we replace $\mu$ by $\lambda \mu$, the inclusion in $\X(\bS^2, \mu)$ does not.)
 
 The pullback of the invariant bilinear form \[\kappa_{\infty}(f,g) = \frac{12\pi}{\vol_{\mu}(\bS^2)^3}\int_{\bS^2}fg\mu\]
along the inclusion $\iota \colon \mathfrak{so}(3)  \hookrightarrow C^{\infty}_0(\bS^2)$ is a multiple of the Killing form $\kappa_{\su(2)}$, which we 
normalise by $\kappa_{\su(2)}(h,h) = -2$ on the coroot $h = \mathrm{diag}(i, -i)$.
To determine the constant of proportionality, let $X, Y, Z$ be the standard basis of $\mathfrak{so}(3) \simeq \su(2)$, satisfying $[X,Y] = Z$, 
$[Y,Z] = X$, and $[Z, X] = Y$. Then $Z = h/2$, so $\kappa_{\su(2)}(Z, Z) = -1/2$, and by conjugation invariance we find  
\begin{equation}\label{eqSumsToOne1}
\kappa_{\su(2)}(X, X) + \kappa_{\su(2)}(Y, Y) + \kappa_{\su(2)}(Z, Z) = -3/2.
\end{equation}

Since $\kappa_{\infty}$ does not depend on the scaling of the $\SU(2)$-invariant volume form $\mu$, we 
choose the area on the unit sphere $\bS^2$ that is induced by the Euclidean metric, $\mu = \sin(\theta)d\phi\wedge d\theta$.
Then $\iota(Z) = \partial_{\phi}$, and since 
$i_{\partial_{\phi}}\mu = -d\cos(\theta)$, this is a Hamiltonian vector field $\partial_{\phi} = X_{z}$ with potential $z= \cos(\theta)$.
Since $x^2 + y^2 + z^2 = 1$ on $\bS^2$, we find
\begin{equation}\label{eqSumsToOne2}
\kappa_{\infty}(x,x) + \kappa_{\infty}(y,y) + \kappa_{\infty}(z,z) = \frac{12\pi}{\vol_{\mu}(\bS^2)^3}\int_{\bS^2} (x^2 + y^2 + z^2) \omega = \frac{3}{4\pi}.
\end{equation}
Comparing \eqref{eqSumsToOne1} to \eqref{eqSumsToOne2}, we find $\iota^* \kappa_{\infty} = -\frac{1}{2\pi} \kappa_{\su(2)}$.
For $\phi, \chi \in C^{\infty}(\bS^1, \R)$ and $X, Y \in \su(2)$, we then find
\begin{eqnarray*}
(L\iota)^*\psi_{\infty} (\phi X, \chi Y) &=& \int_{\bS^1}\phi \chi'  \kappa_{\infty}(\iota(X), \iota(Y)) dt\\
&=& -\frac{1}{2\pi}\int_{\bS^1}\phi \chi'  \kappa_{\mathfrak{su}(2)}(X, Y) dt\\
&=& -\psi_{\mathfrak{su}(2)}(\phi X, \chi Y),
\end{eqnarray*}
so 
\begin{equation} \label{eq:Compatible2withInfty}
(L\iota)^* \psi_{\infty} = -\psi_{\mathfrak{su}(2)}.
\end{equation}
 \end{proof}

\subsubsection{Projective unitary representations}

Our interest in integrability conditions is motivated by the projective unitary representation theory of $L\fk$.
Let $G$ be a locally convex Lie group and let $\overline{\rho} \colon G \rightarrow \mathrm{U}(\cH)$ be a projective unitary representation.
Then a vector $\psi \in \cH$ is called \emph{smooth} if the orbit map 
\[
	G \rightarrow \mathbb{P}(\cH)\colon g \mapsto \overline{\rho}(g)\psi
\]
is smooth, and we denote the space of smooth vectors by $\cH^{\infty} \subseteq \cH$. The projective representation is called 
\emph{smooth} if $\cH^{\infty}$ is dense in $\cH$ for the Hilbert space topology. 

For finite dimensional Lie groups this is equivalent to continuity in the strong operator topology, but for  
infinite dimensional Lie groups this is a sslightly stronger requirement, which enables one to define a Lie algebra representation 
of $\mathrm{Lie}(G)$ on $\cH^{\infty}$.
 
 Every continuous projective unitary representation of $G$ gives rise to a central extension 
 $G^{\sharp} = \{(g, U) \in G \times \mathrm{U}(\cH) \,;\, \overline{\rho}(g) = \overline{U}\}$, which is a central extension 
 of topological groups $\mathrm{U}(1) \rightarrow G^{\sharp} \rightarrow G$.
By \cite{JN19}, this extension is a locally convex Lie group extension if $\overline{\rho}$
is a smooth projective unitary representation, and the corresponding Lie algebra 
extension $\R \rightarrow \fg^{\sharp} \rightarrow \fg$ is characterized up to (strict) isomorphism 
by the cohomology class $[\psi] \in H^2(\fg, \R)$ of the curvature cocycle 
$\psi(x, y) := [\sigma(x), \sigma(y)] - \sigma([x,y])$ of a continuous linear splitting $\sigma \colon \fg \rightarrow \fg^{\sharp}$.

\begin{Corollary}
If $(\oline{\rho}, \cH)$ is a smooth projective unitary representation of the locally convex Lie group $L\SDiff(\bS^2, \mu)$, 
then the derived projective Lie algebra representation $(d\rho, \cH^{\infty})$ of $LC^{\infty}_{0}(\bS^2)$
satisfies \[[d\rho(F), d\rho(G)]  = d\rho(\{F,G\}) + ic_{\infty} \psi(F, G)\one,\]
where $\psi$ is cohomologous to \eqref{eq:LKcocycle}, and where $c_{\infty}\in \Z$.
\end{Corollary}
 
 \section{Twisted loop algebras}\label{sec:3twisted}

The action of $\mathrm{SO}(2)$ on $\bS^1$ by rotation gives rise to a semidirect product 
\[
\mathrm{SO}(2)\ltimes L\SDiff(\bS^2),
\]
where $\mathrm{SO}(2)$ acts on $L\SDiff(\bS^2)$ by conjugation,
$T_{\tau} \phi_{t} T_{-\tau}(x) = \phi_{t - \tau}(x)$. At the Lie algebra level 
we find one extra generator $\td$ that acts by differentiation in the $\bS^1$-direction,
\[
	\R \td \ltimes LC^{\infty}_{0}(\bS^2)
\]
with 
\[[\td, F]_{t}(x) = \frac{d}{dt}F_{t}(x).\]
It seems natural to interprets $\bS^1$ as the time direction, and
$\td$ as the generator of time translations -- which becomes the Hamiltonian when the Lie algebra is represented 
on a Hilbert space.

However, at least in 
certain contexts, this picture has to be slightly adjusted. 
The $3+1$ dimensional Anti de Sitter space $\mathrm{AdS}_{3,1}$ can be realised inside $\R^{5}$ as
\[
\mathrm{AdS}_{3,1} = \{(x,y,z,u,v) \in \R^5\,;\, x^2 + y^2 + z^2 - u^2 - v^2 = - L^2\},
\]
with Lorentzian metric inherited from the metric with signature $(+,+,+,-,-)$ on $\R^5$.
Its isometry group $\mathrm{SO}(3,2)$ acts by conformal transformations on the boundary $Q_{2,1}$ of the 
conformal completion of $\mathrm{AdS}_{3,1}$, 
the projective quadric 
\[
Q_{2,1} = \{[x,y,z,u,v] \in \mathbb{P}(\R^5) \,;\, x^2 + y^2 + z^2 - u^2 - v^2 = 0\}.
\]
In particular $\mathrm{SO}(3)$ acts on the variables $(x,y,z)$ and $\mathrm{SO}(2)$ on the variables $(u,v)$.
Since $x^2 + y^2 + z^2 = u^2 + v^2$, there exist precisely two representatives 
of the ray $[x,y,z,u,v] \in \mathbb{P}(\R^5)$ that satisfy $x^2 + y^2 + z^2 = 1 = u^2 + v^2$.
So $Q_{3,1}$ is diffeomorphic to  $(\bS^1 \times \bS^2)/\Z_2$, with the $\Z_2$-action 
$(t, x) \mapsto (-t, -x)$.
This is a sphere bundle over $\bS^1/\Z_2$. 
In the following it will be convenient to distinguish the circle $\bS^1/\Z_2$ from its twofold cover 
$\bS^1 \subset \C$, and write $x$ for $(x,y,z) \in \bS^2$ and 
$z = e^{2\pi it}$ for $z = u + iv$ in $\bS^1$ if both are normalised to one. The
bundle map $Q_{3,1} \rightarrow \bS^1/\Z_2$ then becomes
\[
\pi\colon Q_{3,1} \rightarrow \bS^1/\Z_2, \quad \pi([\vec{x},e^{2\pi it}]) = e^{4\pi it}
\]
if we identify $\bS^1/\Z_2$ with the circle in the complex plane.
So a full rotation in $\mathrm{SO}(2)$ traverses the circle $\bS^1/\Z_2$ \emph{twice},
and 
the action of $\mathrm{SO}(2)$ on $Q_{3,1}$ gives rise to an Ehresmann connection on $Q_{3,1} \rightarrow \bS^1/\Z_2$
that has nontrivial holonomy, namely the inversion $P(x)= -x$ on the sphere $\bS^2$.

\subsection{The structure of twisted loop algebras}

 To describe the loop groups in this setting, it is convenient to consider the (slightly) more general case of a smooth sphere bundle 
 $\pi \colon F \rightarrow \bS^1$.
Then $F$ is isomorphic to $F_{\Phi}:= (\R \times \bS^2)/ \Z$, the sphere bundle that arises for a $\Z$-action $(\tau, x) \mapsto (\tau + n, \Phi^{-n}(x))$ on $\R \times \bS^2$, determined by a diffeomorphism 
$\Phi$ of $\bS^2$. 
We will assume that either $\Phi^*\mu = \mu$ or $\Phi^*\mu = -\mu$. 
An example of the first case is the identity, an example of the second case is the inversion $P(x) = -x$.
If $\Phi^*\mu = \mu$, then 
every fibre $F_t$ comes equipped with a volume form $\mu_t$. If $\Phi^*\mu = -\mu$, then this volume form $\mu_t$ is well-defined 
only up to a sign. 

For this reason we replace the volume form $\mu \in \Omega^2(\bS^2)$ by a density $|\mu| \in \Gamma(\bS^2, |\Lambda|^{-2}T\bS^2)$, 
for example by tensoring $\mu$ with an orientation class. This has the advantage of slightly enlarging the group of automorphisms 
(the inversion does satisfy ${P^*|\mu| = |\mu|}$) while leaving the Lie algebra of infinitesimal automorphisms unchanged, $\X(\bS^2, \mu) = \X(\bS^2, |\mu|)$.
We define the  \emph{twisted loop group} as follows.

\begin{Definition}
The twisted loop group $L_{\Phi}\SDiff(\bS^2, \mu)$ is the group of all vertical automorphisms $\phi$ of the sphere bundle 
$F_{\Phi} \rightarrow \bS^1$ such that for each $t\in \bS^1$, the map $\phi_{t} \colon F_t \rightarrow F_t$ on the fibre
$F_{t} \simeq \bS^2$
preserves the orientation, and it preserves the density $|\mu|_t$.
\end{Definition}

\paragraph{The untwisted case}
If $\Phi^*\mu = \mu$, then $F \rightarrow \bS^1$ is isomorphic to the trivial sphere bundle.
Indeed, by Smale's theorem \cite{Smale1959} the group $\SDiff(\bS^2, \mu)$ 
has the homotopy type of $\mathrm{SO}(3)$, and in particular it is connected. 
Since $\Phi \in \SDiff(\bS^2, \mu)$, we can choose a smooth path of volume preserving diffeomorphisms 
$\Phi_{t}$ with $\Phi_{0} = \mathrm{Id}_{\bS^2}$ and $\Phi_{t+n} = \Phi^{-n} \circ \Phi_t$, and the map 
$[t, x] \mapsto [t, \Phi_{t}(x)]$ is a well-defined bundle isomorphism $\bS^1 \times \bS^2 \rightarrow F$.
 This means that the group $L_{\Phi}\SDiff(\bS^2,\mu)$ is isomorphic to the untwisted loop group $L\SDiff(\bS^2,\mu)$, 
 but note that the $\mathrm{SO}(2)$-action on the trivial bundle does not coincide with the $\R$-action on $(\R \times \bS^2)/\Z$.
 The generator of the $\R$ action is not $\td = \frac{d}{dt}$, but 
 \[
 \td + A
 \]
 with $A_t = (\frac{d}{dt}\Phi^{-1}_{t}) \circ \Phi_{t}$ the (left) logarithmic derivative of the path $\Phi_t$.
 (One can think of $A$ as an Ehresmann connection on $F\rightarrow \bS^1$, with holonomy governed by $\Phi_t$.)
 Since $A \in LC^{\infty}_0(\bS^2)$, the Lie algebra of the semidirect product 
 $\R \ltimes L_{\Phi}\SDiff(\bS^2,\mu)$ is isomorphic to the untwisted Lie algebra 
 $\R \td \ltimes LC^{\infty}_0(\bS^2)$, but the generator of the $\R$-action is shifted by 
 $A\in LC^{\infty}_0(\bS^2)$.

\paragraph{The twisted case}
The main example for the twisted case is $\Phi = P$, the inversion of the sphere, $P(x) = -x$. Indeed, 
any two diffeomorphisms $\Phi, \Phi'$ with $\Phi^*\mu = (\Phi')^*\mu = -\mu$ differ by an element $\Phi'\Phi^{-1} \in \SDiff(\bS^2,\mu)$, 
so as before we can choose a smooth path $\Phi_t$ in $\SDiff(\bS^2,\mu)$ with $\Phi'\Phi_{t+1}\Phi^{-1} = \Phi_t$ for all $t\in \R$ and 
$\Phi_0 = \mathrm{Id}_{\bS^2}$. This implements an isomorphism $F_{\Phi} \rightarrow F_{\Phi'}$ of sphere bundles by $[t,x] \mapsto [t, \Phi_{t}(x)]$. 
As before this induces an isomorphism of Lie algebras, so we may assume $\Phi = P$ without loss of generality in the orientation-reversing case. 
Note however that the generator of the infinitesimal $\R$-action $\frac{d}{d\tau}$ from $F_{\Phi}$ is 
represented in the model for $F_P$ by
\[
{\textstyle \frac{d}{d\tau}} + A,
\]
where again the left logarithmic derivative $A_t$ of $\Phi_t$ can be interpreted as an Ehresmann connection on $F_P\rightarrow \bS^1$.

The twisted loop group $L_P\SDiff(\bS^2, \mu)$ consists of smooth functions $\phi \colon \R \rightarrow \SDiff(\bS^2, \mu)$
with $\phi_{t+1} = P^{-1} \circ \phi_{t} \circ P$. 
Equivalently, it is isomorphic to the subgroup $L\SDiff(\bS^2,\mu)^{\Z_2}$ that is invariant under 
the $\Z_2$-action $(e^{it}, x) \mapsto (-e^{it}, -x)$. In this identification, the circle $\bS^1$ in $L\SDiff(\bS^2,\mu)$ 
acts as a twofold cover of the circle $\R/\Z$ in $L_P\SDiff(\bS^2, \mu)$.
The Lie algebra $L_P\fk$ of $L_P\SDiff(\bS^2, \mu)$ consists of smooth functions 
$\xi \colon \R \rightarrow \X(\bS^2, \mu)$ such that $\xi_{\tau + 1} = TP^{-1} \circ \xi_{\tau} \circ P$.

We can identify $\X(\bS^2, \mu)$ with $C^{\infty}_0(\bS^2)$, but if we do so, we need to keep in mind that the identification 
depends on $\mu$. If $X_f$ is the hamiltonian vector field with potential $f$, then 
$\Phi_* X_{f} := T\Phi^{-1} \circ X_{f} \circ \Phi$ is the hamiltonian vector field with potential $\Phi^*f$ if $\Phi$ respects the orientation, 
and $-\Phi^*f$ if $\Phi$ reverses the orientation. In particular, the inversion $P(x) = -x$ gives rise to the 
automorphism 
\begin{equation}
\alpha_{P}(f) = -P^*f
\end{equation}
 that satisfies $\alpha_P(q) = (-1)^{\mathrm{deg}(q)+1}$ on homogeneous 
polynomials $q$ of degree $\mathrm{deg}(q)$. 
In this picture the twisted loop algebra can be identified with
\[
L_{P}\fk = C^{\infty}_{0}(\bS^1\times \bS^2)^{\Z_2},
\]
the smooth functions $F \colon \bS^1 \times \bS^2 \rightarrow \R$ that integrate to zero on every 2-sphere $\{z\} \times \bS^2$, 
and satisfy $F(-z,-x) = - F(z, x)$. 
The generator of the $\mathfrak{so}(2)$-action is 
$\td F(e^{it}, x) = \frac{\partial}{\partial t} F(e^{2\pi it}, x)$. Note that this corresponds to $\td \xi = 2\frac{d}{d\tau}\xi$ in the 
picture where sections are represented by 
$(\R \times \X(\bS^2, |\mu|))^{\Z}$. Indeed, the base of the fibre bundle is $\bS^1/\Z_2$ in the first and 
$\R/\Z$ in the second picture.

\paragraph{An algebraic model}

The inversion of the sphere gives an automorphism 
$\alpha_{P}$ of the Poisson algebra $\fk = C^{\infty}_{0}(\bS^2)$ that is of order two,
that is, a $\Z^2$-grading where $\fk[0]$ is even and $\fk[1]$ is odd under the inversion. 
Since the automorphism is $\mathrm{SO}(3)$-invariant, the grading is invariant as well.
The even part $\fk[0]$ contains the homogeneous polynomials of odd degree, 
and the odd part $\fk[1]$ contains the homogeneous polynomials of even degree except the constants.
The complexification ${(L_P\fk) \otimes_{\R} \C}$ contains the dense subalgebra 
\[
\widehat{\fk}^{P}_{\C} = \C \td \ltimes \bigoplus_{\substack{n, k \in \Z\\k >0\\ n+k \;\mathrm{odd}}} z^n \otimes V_{k},
\]
where $V_k$ are the harmonic homogeneous polynomials of degree $k$.
The bracket is graded in $n$, but only filtered in $k' := k-1$. In particular the subspace $\one \otimes V_1$ 
is a Lie subalgebra, canonically isomorphic to $\mathfrak{so}(3)$.
%

\subsection{Twisted current algebras}

Although our motivation comes from sphere bundles over the circle $\bS^1$, it is convenient for the 
classification of 2-cocycles to consider locally trivial 2-sphere bundles $\pi \colon F \rightarrow M$
over a (not necessarily compact) manifold $M$.
We call ${F \rightarrow M}$ an \emph{unoriented symplectic bundle} if it can be covered by local trivialisations 
$\Phi_i \colon \pi^{-1}(U_i) \rightarrow U_i \times \bS^2$
over $\{U_i\}_{i\in \mathcal{I}}$ 
for which the transition functions $g_{ji} \colon U_{j} \cap U_{i} \rightarrow \mathrm{Diff}(\bS^2)$
take values in 
\[
\SDiff(\bS^2, |\mu|) = \{\phi \in \Diff(\bS^2)\,;\, \phi^* \mu = \pm \mu\}.
\]

This is a slight relaxation of the notion of a symplectic bundle, where the structure group would be reduced to $\SDiff(\bS^2, \mu)$.
From the density $|\mu| \in \Gamma(\bS^2, |\Lambda|^{-2}TM)$ on the 2-sphere, we obtain a 
\emph{vertical density} $|\nu|$ on $F$ that is equivalent to $|\mu|$ on every fibre.
Explicitly, let $T^vF$ be the vertical tangent bundle, and let $|\Lambda|^{-2}T^vF$ be the 
real line bundle associated to the principal $\mathrm{Gl}(2,\R)$-bundle $\mathrm{Fr}(T^vF)$ of vertical frames along the 
character $\chi(g) = |\mathrm{det}(g)|^{-1}$. 
Since the structure group preserves the density $|\mu|$ on $\bS^2$, 
we obtain a well-defined global section $|\nu| \in \Gamma(F, |\Lambda|^{-2}T^vF)$
by setting $|\nu| = \Phi^*_i p_2^* |\mu|$ in local coordinates, where $p_2 \colon U_i \times \bS^2 \rightarrow \bS^2$
is the projection on the second factor.

The triple $(F, M, |\nu|)$ gives rise to the Lie algebra 
\[
\fg_{c}(M) = \{X \in \Gamma_{c}(F, T^vF) \,;  L_{X} |\nu| = 0 \}.
\] 
Alternatively, since $\SDiff(\bS^2, |\mu|)$ acts by automorphisms on $\fk = \X(\bS^2, \mu)$, we can also 
view 
\[\fg_{c}(M) = \Gamma_{c}(M, \mathfrak{K})
\] 
as the Lie algebra of compactly supported sections of a Lie algebra bundle 
$\mathfrak{K} \rightarrow M$ with typical fibre $\fk = \X(\bS^2, \mu)$.
In this picture, the fibre over $x\in M$ is the Lie algebra $\mathfrak{K}_{x} = \X(F_x, |\nu|)$.

For $M = \bS^1$ we of course recover the twisted loop algebras:
the sphere bundle $F = {(\R \times \bS^2)/\Z}$ over $\bS^1$ that arises from the $\Z$-action $(\tau, x) \mapsto (t+n, \Phi^{-n}(x))$ 
inherits the vertical density $p^*|\mu|$ from $\R \times \bS^2$ (because $\Phi^*|\mu| = |\mu|$), 
and the Lie algebra $\fg(\bS^1)$ is isomorphic to $L_{\Phi}\fk$. 

For the trivial bundle $F = M \times \bS^2$, we recover the Lie algebra
$\fg_c(M) = C^{\infty}_c(M, \fk)$ of compactly supported smooth functions with values in the Fr\'echet--Lie algebra 
$\fk = \X(\bS^2, |\mu|)$. 
In general the map $U \mapsto \fg_{c}(U)$ on open subsets of $M$
 is a flabby cosheaf (in the sense of \cite{Bredon1968})  of Lie algebras.
The local trivialisations give rise to Lie algebra isomorphisms $\fg_{c}(U) \simeq C^{\infty}_{c}(U, \fk)$ 
for sufficiently small open subsets $U \subseteq M$.


If we equip $\bS^2$ with a density $|\mu|$ rather than a volume form $\mu$, we have to slightly revise the notion of a potential.
Indeed, since $\X(\bS^2, \mu) = \X(\bS^2, -\mu)$ we can simply write $\X(\bS^2, |\mu|)$ for the Lie algebra of vector fields
which preserves $|\mu|$ with either orientation. But although both Lie algebras are isomorphic to $C^{\infty}_0(\bS^2)$, 
the isomorphism depends on the choice of orientation: if a vector field $X$ has potential $f$ with respect to $\mu$, 
$i_{X}\mu = df$, then the same vector field $X$ has potential $-f$ with respect to $-\mu$.
Rather than functions on $\bS^2$, we therefore think of the potentials as sections of the real line bundle $\mathcal{L} \rightarrow \bS^2$
defined by $\mathcal{L} := \mathrm{Fr}(\bS^2)\times_{\sigma}\R$, 
where $\sigma \colon \mathrm{Gl}(2,\R) \rightarrow \{\pm 1\}$ is the sign of the determinant.
Since $\mathcal{L} \otimes \mathcal{L}$ is trivial, $fg$ is an honest function, and 
\[
\kappa(X_f, X_{g}) = \int_{\bS^2} fg |\mu|
\]
is an invariant bilinear form on $\X(\bS^2, |\mu|)$. From this description it is clear that $\kappa$ is invariant not just under $\SDiff(\bS^2, \mu)$, but 
even under the group $\SDiff(\bS^2, |\mu|)$
of (not necessarily orientation preserving) automorphisms of $(\bS^2, |\mu|)$.
By contrast, the Poisson bracket $\{f,g\}$ is not as a function on $\bS^2$ but a section of $\mathcal{L} \rightarrow \bS^2$, 
namely the unique potential for the vector field $[X_f, X_g]$.
The formula $L_{X_f}g = \{f,g\}$ remains true if both sides are interpreted as sections of $\mathcal{L}$.


If we define $\mathcal{L} \rightarrow F$ as the real line bundle $\mathcal{L} := \mathrm{Fr}(T^vF)\times_{\sigma}\R$
associated to the vertical frame bundle along the sign of the determinant, then every $X\in \fg_{c}(M)$ gives rise to a potential 
$F\in \Gamma_{c}(F, \mathcal{L})$ with $i_{X}|\nu| = - d^{v}F$, where $d^v$ denotes the `vertical' de Rham differential on the fibres of $F$.
This is well defined because a choice of orientation changes signs on both sides. 
Since $FG$ is an honest function, fibre integration yields  a bilinear form 
\[
\kappa(X_{F},Y_{G}) = \fint FG |\nu|
\]
on $\fg_c(M)$ with values in $C^{\infty}_{c}(M)$. It is invariant in the sense that 
\[\kappa([X_{F}, X_{G}], X_H) + \kappa(X_G, [X_{F}, X_{H}]) = 0.\]

Finally, let $\lambda$ be a closed current on $M$, i.e.\ a continuous linear functional $\lambda \colon \Omega^1_{c}(M) \rightarrow \R$
that vanishes on the exact 1-forms $d\Omega^0_{c}(M)$. 
Let $\nabla$ be an Ehresmann connection on $F \rightarrow M$ that is compatible with $|\nu|$, 
in the sense that the holonomy along a path $\gamma$ from $x \in M$ to $y\in M$ 
induces a diffeomorphism $\Gamma(\gamma) \colon F_x \rightarrow F_y$ that respects the vertical density, 
$\Gamma(\gamma)^* |\nu|_{F_y} = |\nu|_{F_x}$.
Then the fibre integral $ \fint F \nabla G |\nu|$ is a compactly supported 1-form on $M$, and
\[
\psi_{\lambda, \nabla}(X_{F}, Y_{G}) = \lambda \Big( \fint F \nabla G |\nu|\Big)
\]
is a 2-cocycle on $\fg_c(M)$.
The cohomology class of $\psi_{\lambda, \nabla}$ does not depend on the choice of connection. Indeed, any other connection $\nabla'$
differs from $\nabla$ by a 1-form $A \in \Omega^1(F, T^vF)$ that is horizontal, and has the property that for every 
$v_x\in T_xM$, the vector field 
$A(v_x) \colon f \mapsto A_f(v_x)$ on $\pi^{-1}(x)$ preserves the vertical density. So 
$A(v_x) \in \X(\pi^{-1}(x), |\nu|)$ has a potential $a(v_x)$, a section of $\mathcal{L}$ such that 
$i_{X_{A(v_x)}}|\nu| = - d^v a(v)$.
It follows that for $F, G \in \Gamma_{c}(F, \mathcal{L})$, we have
\begin{eqnarray*}
\psi_{\lambda, \nabla'}(X_{F}, X_{G}) - \psi_{\lambda, \nabla}(X_{F}, X_{G})
&=& \lambda \Big( \fint F L_{A} G |\nu|\Big)\\
&=& \lambda \Big(\fint \{F, G\} a |\nu|\Big)
\end{eqnarray*}
is the coboundary of $\eta(X_{F}) = \lambda(\fint F a |\nu|)$.
\begin{Theorem}\label{ThmCocycleTwistedBundles} The continuous cohomology $H^2(\fg_{c}(M), \R)$ is isomorphic to 
the continuous dual $(\Omega^1_{c}(M)/d\Omega^0_{c}(M))'$. Every continuous 2-cocycle on $\fg_c(M)$ is cohomologous to 
\begin{equation}\label{eq:cocycleLambda}
\psi_{\lambda, \nabla}(X_F,X_G) = \lambda\Big( \fint F\nabla G |\nu| \Big)
\end{equation}
for a closed current $\lambda \colon \Omega^1_{c}(M) \rightarrow \R$. 
\end{Theorem}
Before we proceed with the proof (Sections~\ref{sec:CompactSupport} and \ref{sec:sheafify}), let us spell out 
what this means for the twisted loop algebra $L_{P}\fk = C^{\infty}_{0}(\bS^1\times \bS^2)^{\Z_2}$.
To classify the 2-cocycles, we  consider the sphere 
bundle $F = (\bS^1 \times \bS^2)/ \Z_2 \rightarrow \bS^2/\Z_2$, where the action of $\Z_2$ on $\bS^1 \times \bS^2$
is given by $(z, x) \mapsto (-z, -x)$, and the action on $\bS^1$ by $z \mapsto -z$. 
A section of $\mathcal{L} \rightarrow F$ is then precisely a function $F$ on the double cover $\bS^1 \times \bS^2$
that satisfies $F(-z, -x) = - F(z, x)$.
Of course $\bS^1/\Z_2$ is again a circle $\bS^1$, 
and since $\Omega^1(\bS^1)/d\Omega^0(\bS^1)$ is 1-dimensional with generator $\lambda(\alpha) = \oint \alpha$, 
we recover the following result.
\begin{Theorem}\label{ThmCocycleTwistedLoop} The continuous cohomology $H^2(L_{P}\fk, \R)$ is 
1-dimensional. Up to scaling, every continuous 2-cocycle on $L_{P}\fk = C^{\infty}_0(\bS^1 \times \bS^2)^{\Z_2}$
 is cohomologous to 
\begin{equation}\label{eq:cocycleTwist}
\psi^P_{\infty}(F,G) :=  \frac{6\pi}{\vol_{\mu}(\bS^2)^3}\int_{\bS^1}\Big(\int_{\bS^2} F \partial_tG \mu\Big)  dt.
\end{equation}
The cocycle $c^P_{\infty} \psi^{P}_{\infty}$ integrates to the group level if and only if $c^P_{\infty}\in \Z$.
\end{Theorem}
For the integrality condition, note that it follows from Lemma 3.10 in \cite{NW09} that for $M = \bS^1$, 
the integrability 
of the cocycle \eqref{ThmCocycleTwistedBundles} does not depend on the `twist' $\Phi$. The integrability 
condition should therefore be the same as in Theorem~\ref{Thm:integralityCocycles}.
Since the base manifold for the sphere bundle $(\bS^1 \times \bS^2)/\Z_2$ is 
$\bS^1/\Z_2$ rather than $\bS^1$, the cocycle \eqref{eq:cocycleTwist} corresponds to \eqref{eq:cocycleLambda}
for $\lambda(\alpha) = \frac{12\pi}{\vol(\bS^2)^3}\oint \alpha$, the same current as for the cocycle \eqref{eq:cocycleTheoremA}. This explains why the factor 
$6\pi$ (rather than $12\pi$) occurs in \eqref{eq:cocycleTwist}.

\subsection{Classification of 2-cocycles for $C^{\infty}_{c}(M, \fk)$}\label{sec:CompactSupport}
The first step in the proof of Theorem~\ref{ThmCocycleTwistedBundles} is to consider the \emph{local} case. Since the sphere bundle $F \rightarrow M$
trivialises over sufficiently small open subsets $U \subseteq M$, 
the Lie algebra $\fg_{c}(U)$ restricted to these subsets is isomorphic to $C^{\infty}_{c}(U, \fk)$.
In this section we classify the continuous 2-cocycles on these Lie algebras.
In the second step of the proof, in Section~\ref{sec:sheafify}, we reduce the calculation of second Lie algebra cohomology to the local setting, 
closely following \cite{JW13}.

%
In the following we consider $C^{\infty}_{c}(M, \fk)$ as an LF-Lie algebra, with the 
topology of uniform convergence in $\fk$ of all derivatives on compact subsets of $M$.

\begin{Lemma}\label{Thm:OnManifolds}
Let $M$ be a smooth manifold, and let $\fk = C^{\infty}_{0}(\bS^2)$. Then every 2-cocycle on $C^{\infty}_{c}(M, \fk)$ 
is cohomologous to $\psi_{\lambda}$, where $\lambda \colon \Omega_{c}^1(M)/d\Omega_{c}^0(M) \rightarrow \R$ is a closed current
and
\[
\psi_{\lambda}(F,G) = \lambda(\kappa(FdG)).
\]
\end{Lemma}
From this we of course immediately recover Theorem~\ref{ThmCocycleLoop}. Indeed, 
for $M = \bS^1$ we have $\Omega_{c}^1(\bS^1)/d\Omega_{c}^0(\bS^1) = H^1_{\mathrm{dR}}(\bS^1, \R)$, so
 every closed current is a real multiple of $\lambda(adt) = \oint a(e^{it})dt$. 
Since $\psi_{\lambda}$ is proportional to the cocycle $\psi_{\infty}$ from \eqref{eq:cocycleTheorem1},
Theorem~\ref{ThmCocycleLoop} follows.


To prove Lemma~\ref{Thm:OnManifolds}, we will need a number of preparations.
We identify $C^{\infty}_{c}(M, \fk)$ with the projective tensor product $A\otimes_{\pi}\fk$, 
where $A = C^{\infty}_{c}(M)$ is the algebra of compactly supported smooth functions.
If $M$ is noncompact, then the ring $A = C^{\infty}_{c}(M)$ does not have a unit, so we cannot apply the results from
 \cite{NW08} directly. We first solve this by showing that appending a unit to $A$ leaves the situation essentially unchanged.
Indeed, adjoining the identity to $A = C^{\infty}_{c}(M)$ yields
the unital algebra $A_{\one} = C^{\infty}_{c}(M) \oplus \R \one$, and we now show that 
 $A_{\one} \otimes_{\pi}\fk = C^{\infty}_{c}(M, \fk) \rtimes \fk$ has the same continuous 
 second cohomology as $A\otimes_{\pi} \fk = C^{\infty}_{c}(M, \fk)$.

\begin{Lemma}\label{Lemma:appendIdentity}
If $M$ is noncompact, the inclusion $\iota \colon C^{\infty}_{c}(M, \fk) \hookrightarrow {C^{\infty}_{c}(M, \fk) \rtimes \fk}$
yields an isomorphism $\iota^* \colon H^2(C^{\infty}_{c}(M, \fk) \rtimes \fk, \R) \rightarrow H^2(C^{\infty}_{c}(M, \fk), \R)$ in continuous
second Lie algebra cohomology. 
\end{Lemma}
\begin{proof}
This essentially follows from the proof of \cite[Theorem~2.7]{JW13}. Since we require a few minor adaptations, 
we repeat the argument here for the convenience of the reader. 

Since $\fk = C^{\infty}_{0}(\bS^2)$ is perfect \cite{ALDM74}, the Lie algebra $A\otimes_{\pi}\fk$ is topologically perfect 
by \cite[Prop.~II.4]{JW13}. It follows that 
 every continuous 2-cocycle \[\psi \colon {A\otimes_{\pi}\fk \times A\otimes_{\pi}\fk} \rightarrow \R\]
 is \emph{diagonal} by \cite[Lemma~2.2]{JW13}, so that
 $\mathrm{supp}(a) \cap \mathrm{supp}(b) = \emptyset$ implies ${\psi(a\otimes f, b\otimes g)} = 0$.
To show that $\iota^*$ is surjective, let 
$\psi \colon {A\otimes_{\pi}\fk \times A\otimes_{\pi}\fk} \rightarrow \R$ be a continuous cocycle.
We extend $\psi$ to a cocycle $\psi'$ on $A_{\one}\otimes_{\pi}\fk$ by setting $\psi'(\one \otimes f, \one \otimes g) := 0$,
and by defining $\psi'(a \otimes f, \one \otimes g) := \psi(a \otimes f, \chi \otimes g)$ for some function $\chi \in C^{\infty}_{c}(M,\R)$
that satisfies $\chi|_{\mathrm{supp}(a)} = 1$. Since $\psi$ is diagonal, this expression does not depend on the choice of $\chi$.
To check that $\psi'$ is again a cocycle, note that 
\[
\psi'(a\otimes f, [\one \otimes g, \one \otimes h]) + \mathrm{cyclic} = 
\psi'(a\otimes f, [\chi \otimes g, \chi' \otimes h]) + \mathrm{cyclic} = 0
\] 
for $\chi, \chi' \in A$ where $\chi$ is equal to $1$ on the support of $a$, and $\chi'$ is equal to $1$ on the support of $\chi$.
This shows the cocycle identity on $(A\otimes_{\pi} \fk) \times \fk \times \fk$. The 
cocycle identity on $\fk \times (A\otimes_{\pi} \fk) \times (A \otimes_{\pi} \fk)$ is similar, and the cocycle identity on $\fk \times \fk \times \fk$ is trivial.

To show that $\iota^*$ is injective, let $\psi'$ be a continuous cocycle on $A_{\one} \otimes_{\pi} \fk$ for which $\iota^*\psi' = d \gamma$ 
is the coboundary of a continuous linear functional $\gamma$ on $A\otimes_{\pi} \fk$.
If we extend $\gamma$ by zero to $\gamma' \colon A_{\one} \otimes_{\pi}\fk \rightarrow \R$,
then $\psi'' := \psi - d\gamma'$ is zero on $A\otimes_{\pi}\fk \times A \otimes_{\pi} \fk$. So
\[
\psi''([a\otimes f, b\otimes g], \one \otimes h) + \mathrm{cyclic} = 0
\]
by the cocycle identity, and since $\psi''$ vanishes on $A\otimes_{\pi}\fk \times A \otimes_{\pi} \fk$ this implies 
that $\psi''([A\otimes \fk, A\otimes \fk], \one \otimes \fk) = \{0\}$.
Since $\fk$ is perfect and $A^2 = A$, this implies that $\psi''(A\otimes \fk, \fk) = \{0\}$.
Since $\psi''$ is continuous and $A\otimes \fk \subseteq A\otimes_{\pi}\fk$ is dense 
we have $\psi''(A\otimes_{\pi} \fk, \fk) = \{0\}$.
So $\psi''$ factors through a continuous cocycle on $\fk$, 
which is a coboundary of a continuous 1-cochain by Theorem~\ref{Thm:H2Poisson}.
\end{proof}

\begin{proof}[Proof of Lemma~\ref{Thm:OnManifolds}]
Let $\psi$ be a continuous cocycle on $A\otimes_{\pi}\fk$, where $A = C^{\infty}_{c}(M)$ and 
$\fk$ is the Poisson algebra $\fk = C^{\infty}_{0}(\bS^2)$.
If $M$ is compact we apply Lemma~\ref{Lem:NeebWagemannRefined} to the unital algebra $A$.
If $M$ is noncompact, then $\psi$ extends to a continuous cocycle on $A_{\one}\otimes_{\pi}\fk$
by Lemma~\ref{Lemma:appendIdentity}, and we apply Lemma~\ref{Lem:NeebWagemannRefined} to the unital algebra
$A_{\one} = {C^{\infty}_{c}(M) \oplus \R \one}$.

Either way, we find that $\psi(a\otimes f, b\otimes g) = \lambda(adb)\kappa(f,g) + \gamma(ab \otimes [f,g])$
for a linear functional $\lambda \colon \Omega^1_{A}/dA \rightarrow \R$, and a 
1-cochain $\gamma \colon A\otimes \fk \rightarrow \R$ that is continuous in the $\fk$-variables. 
Recall that the universal property of the K\"ahler differentials yields a canonical map 
$\Omega^1_{A}/dA \rightarrow \Omega_c^1(M)/d_{\mathrm{dR}}C^{\infty}_{c}(M)$.
We need to show that $\lambda$ factors through a continuous linear map $\Omega^1_{c}(M)/dC_{c}^{\infty}(M) \rightarrow \R$, 
and that the cocycle $\gamma$ is continuous.

Since $\fk$ is perfect \cite{ALDM74}, it follows from \cite[Prop.~II.4, Lemma~2.2]{JW13}
that $\psi$ is diagonal: $\mathrm{supp}(a) \cap \mathrm{supp}(b) = \emptyset$
implies $\lambda(adb) = 0$. If we choose $f\in \fk$ with $\kappa(f,f) = 1$, we find that 
the bilinear map $(a,b) \mapsto \lambda(adb) = \psi(a\otimes f, b \otimes f)$ 
is diagonal as well.

But then 
the map $\check{\lambda} \colon C^{\infty}_{c}(M) \rightarrow C^{\infty}_{c}(M)'\colon b\mapsto \lambda(\,\cdot\, db)$
from test functions to distributions is \emph{local}, in the sense that $\mathrm{supp}(\check{\lambda}(a)) \subseteq \mathrm{supp}(a)$
for all $a\in C^{\infty}_{c}(M)$. 
Since the cocycle is jointly continuous, the map $(a, b) \mapsto \lambda(adb)$ is jointly continuous as well, and 
$b \mapsto \lambda(\,\cdot\,db)$ is a differential operator by Peetre's Theorem.
It is a differential operator of first order because 
\[(\delta_{\alpha}\delta_{\beta}\check{\lambda})(b)  := \check{\lambda}(\alpha \beta b) - \alpha \check{\lambda}(\beta b) - 
\beta \check{\lambda}(\alpha b) + \alpha \beta\check{\lambda}( b) = 0
\]
for all $\alpha, \beta \in C^{\infty}_{c}(M)$,
so there exist distributions $\phi_0, \phi^{\mu}_1$ on $M$ such that
$\check{\lambda}(b)(a) = b\phi_0 + \partial_{\mu}b \phi^{\mu}_1$ in local coordinates (repeated index indicates summation).
So $\lambda(adb) = \phi_0(ab) + \phi^{\mu}_1(a\partial_{\mu}b)$. 
With $b=\one$ we see that $\phi_0 = 0$, so 
$\lambda$ drops to a continuous linear map 
$\phi_1 \colon \Omega^1_{c}(M)/dC^{\infty}_c(M) \rightarrow \R$ as required.

Finally, to show that the 1-cochain $\gamma$ is continuous we proceed as in \eqref{eq:get1from2}. 
The 1-cochain can be recovered from $\psi$ by 
\[\gamma(a \otimes \{f,g\}) = {\textstyle\frac{1}{2}}(\psi(a\otimes f, \one\otimes g) - \psi(a\otimes g, \one \otimes f)),\]
and since $\psi$ is continuous, the map $a \mapsto \gamma(a\otimes f)$ is continuous as well.
We already knew from Remark~\ref{Rk:partCont1chain} that $f \mapsto \gamma(a\otimes f)$ is continuous, so the bilinear map 
$(a, f) \mapsto \gamma(a\otimes f)$ is separately continuous. 
Since $A$ and $\fk$ are barrelled locally convex 
 vector spaces, the Banach-Steinhaus Theorem implies that the bilinear map is jointly continuous on $A\times \fk$. 
 This means that $\gamma$
 extends to a
a continuous linear functional $\gamma \colon A\otimes_{\pi} \fk \rightarrow \R$ by the universal property of projective tensor product, 
finishing the proof.
\end{proof}

\subsection{The proof of Theorem~\ref{ThmCocycleTwistedBundles}}\label{sec:sheafify}
In the second step, following \cite{JW13}, we reduce the calculation of second Lie algebra cohomology to the local setting. 

\begin{proof}[Proof of Theorem~\ref{ThmCocycleTwistedBundles}]
For an open $U\subseteq M$, let $\fg_{c}(U)$ be the Lie algebra associated to the restriction of $\pi \colon F \rightarrow M$ to $U$.
For sufficiently small open neighbourhoods $U \subseteq \R$, this is isomorphic to $C^{\infty}_{c}(U, \fk)$.
More generally, $\fg_{c}(U)$ is the space of sections of the locally trivial bundle of Fr\'echet--Lie algebras
$\mathfrak{K} \rightarrow M$, where the fibres are given by $\mathfrak{K}_{x} = \X(\pi^{-1}(x), |\nu|)$.

By \cite[Cor.~2.5]{JW13},
the presheaf $\mathcal{S}$ that assigns to an open subset $U\subseteq M$ the second continuous Lie algebra cohomology 
$\mathcal{S}(U) = H^2(\fg_{c}(U), \R)$ satisfies the local identity axiom. 
By \cite[Prop.~2.6]{JW13} the presheaf $\mathcal{F}$ defined by $\mathcal{F}(U) = \mathrm{Hom}(\Omega^1_{c}(U)/dC^{\infty}_{c}(U), \R)$
is a sheaf.
Since 
\[\lambda_{U} \colon \mathcal{F}(U) \rightarrow \mathcal{S}(U); \lambda  \mapsto [\psi_{\lambda, \nabla}]
\] 
is a monomorphism of presheaves that is an isomorphism on sufficiently small open subsets $U$, 
it is an isomorphism by \cite[Prop.~2.8]{JW13}. In particular $\mathcal{S}$ is a sheaf.
Evaluating on $U = M$, we find
\[
H^2(\fg_c(M), \R) \simeq \big(\Omega^1(M)/dC^{\infty}(M)\big)'
\]
as required.
\end{proof}

%

\section{Geometric quantisation}\label{Section:Recap}

In order to exhibit the cocycles \eqref{eq:cocycleTheoremA} on $LC^{\infty}_0(\bS^2)$ 
as limits for $k \rightarrow \infty$ of Kac-Moody cocycles on $L\su(k+1)$, we will need 
a number of well-known results from geometric quantisation. 
In this section we recall these results, and specialise to the 2-sphere. 
As we need to work with the complex structure explicitly, it will be convenient to identify $\bS^2$ with the projective line $\CP^1$.   
 
Let $M = \CP^1$, and let $\mu$ be an $\SU(2)$-invariant symplectic form on $\CP^1$.
We will assume that $\mu$ is positive for the orientation that comes from the complex structure.
Let $\bL =\tau^*$ be the dual of the tautological line bundle $\tau$ over $\CP^1$.
We equip $\bL$ with the $\SU(2)$-invariant connection $\nabla$ and Hermitian form $h$ that are induced by the inclusion of 
$\tau$ in the trivial bundle $\CP^1 \times \C^2$.
Then $(\bL, \nabla, h)$ is an $\SU(2)$-equivariant prequantum line bundle, with integral curvature
\[\textstyle F_{\nabla}  = -i \omega.\]
Every $\SU(2)$-invariant symplectic form $\mu$ on $\CP^1$ is proportional to $\omega$, and as $\vol_{\omega}(\CP^1) = 2\pi$ we have
\begin{equation}
\textstyle \mu = \frac{\vol_{\mu}(\bS^2)}{2\pi}\omega.
\end{equation}

\subsection{Quantomorphism group and Poisson Lie algebra}
Let $\cG = \mathrm{Aut}(\bL, \nabla, h)$ be the \emph{quantomorphism group}, the group of (smooth) automorphisms of the line bundle $\bL \rightarrow \CP^1$ 
that respect the connection and the Hermitian inner product.
An automorphism $g \colon \bL \rightarrow \bL$ covers a diffeomorphism $g_{M} \colon \CP^1 \rightarrow \CP^1$, 
and since $g$ preserves the connection, $g_M$ preserves the symplectic form. 
The map $g \mapsto g_M$ is a surjective group homomorphism from $\cG$ onto $\SDiff(\CP^1)$.
Its kernel, the group $\cZ \simeq \U(1)$ of vertical automorphisms of $(\bL, \nabla, h)$, is the centre of $\cG$.

The quantomorphism group is a Fr\'echet--Lie group \cite{RS81}, and its Lie algebra $\mathrm{Lie}(\cG)$ is isomorphic to the Poisson 
algebra $C^{\infty}(\CP^1)$. 
The isomorphism is perhaps most easily described in terms of the principal $\U(1)$-bundle $P = \{\lambda \in \bL \,;\, h(\lambda, \lambda) = 1\}$, 
and the connection 1-form $\theta \in \Omega^1(P, i\R)$ induced by~$\nabla$.
Then $\mathrm{Lie}(\cG)$ is the Lie algebra of $\U(1)$-equivariant vector fields $X$ on $P$, and 
 the requirement that $L_{X}\theta = 0$  
translates to $i_{X}d\theta + d(i_X\theta) = 0$.
Since both $X$ and $\theta$ are both equivariant, $i_{X}\theta$ is equivariant as well, and \[\textstyle i_{X}\theta = -i \pi^*f\] for some 
$f\in C^{\infty}(\CP^1, \R)$. Since $d\theta = \pi^* F_{\nabla} = -i\pi^*\omega$, the function $f$ determines 
the horizontal component $X_f\in \X(\CP^1, \omega)$ of $X$ by $i_{X_f}\omega + df = 0$.
Since $\CP^1$ is simply connected, every symplectic vector field is Hamiltonian, and the map 
$f \mapsto X_f$ is a surjective Lie algebra homomorphism from $C^{\infty}(\CP^1, \R)$ to the Lie algebra 
$\X(\CP^1, \omega)$ of symplectic vector fields.
Its kernel is the abelian Lie algebra $\R$ of constant functions, and the central extension 
\begin{equation}\label{QuantomorphismCover}
1 \rightarrow \U(1) \rightarrow \cG \rightarrow \SDiff(\CP^1) \rightarrow 1
\end{equation}
of Fr\'echet--Lie groups corresponds to the central extension 
\begin{equation}\label{PoissonCover}
 0 \rightarrow \R \rightarrow C^{\infty}(\CP^1,\R) \rightarrow \X(\CP^1, \omega) \rightarrow 0
\end{equation}
of Fr\'echet--Lie algebras. 

\begin{Remark}
Note that the Lie algebra extension \eqref{PoissonCover} is trivial because the constants split off as an ideal \cite[corr.~3.6]{JV16}. By contrast, 
the group extension \eqref{QuantomorphismCover} is nontrivial;
the natural action of $\SU(2)$ on  $\bL\rightarrow \CP^1$ preserves the connection, 
so it induces an inclusion $\SU(2) \hookrightarrow \cG$ that covers the inclusion $\SU(2)/\Z_2 \hookrightarrow \SDiff(\CP^1)$.
The restriction of \eqref{QuantomorphismCover} to $\SU(2)/\Z_2$ is the $\mathrm{spin}_c$ extension
\[1 \rightarrow \U(1) \rightarrow \SU(2)\times_{\Z_2}\U(1) \rightarrow \SU(2)/\Z_2 \rightarrow 1.\]
Since $\Diff(\CP^1)$ is a principal $\SDiff(\CP^1)$-bundle over the contractible space of smooth positive measures 
of constant mass \cite[Theorem~2.5.3]{Ham82}, and since the inclusion $\SU(2)/\Z_2 \hookrightarrow \SDiff(\CP^1)$
is a homotopy equivalence by Smale's Theorem \cite{Smale1959}, the quantomorphism group $\cG$ has 
the homotopy type of $\mathrm{spin}_c(3) = \SU(2) \times_{\Z_2}\U(1)$.
\end{Remark}

\subsection{Geometric quantisation}

For every $k\in \N$, we denote by $\cH_{k} := \Gamma_{L^2}(\CP^1, \bL^{\otimes k})$ 
the Hilbert space of square integrable sections of $\bL^{\otimes k}$, with inner product
\[
\lra{\sigma_1}{\sigma_2} :=  \frac{1}{2\pi}\int_{\CP^1}h^{\otimes k}(\sigma_1, \sigma_2) \omega.
\] 
The unitary representation of $\cG$ on $\cH_k$ is given by
\[\rho^{k}(g) s := g^{\otimes k} \circ s \circ g_{M}^{-1},\]
and it gives rise to the Lie algebra representation \[
d\rho^k(f) s = \left(ikf - \nabla^{\otimes k}_{X_{f}}\right) s
\]
of $\mathrm{Lie}(\cG) = C^{\infty}(\CP^1)$ on the smooth sections $\Gamma(\CP^1, \bL^{\otimes k})$.

Let $V_k := \Gamma_{\mathrm{hol}}(\CP^1, \bL^{\otimes k})$ be the space of holomorphic sections, and let 
$\Pi_k \colon \cH_k \rightarrow V_k$ be the orthogonal projection onto $V_k$ (the Szeg\H{o} projection). 
The group $\SU(2) \subseteq \cG$ acts on $\bL$ by holomorphic transformations, 
so the $V_k$ are unitary representations of $\SU(2)$. 
By the Borel--Weil Theorem they are irreducible, and all irreducible representations 
of $\SU(2)$ are of this form. To make the isomorphism explicit, we define for every $\alpha \in \C^2{}^*$ the holomorphic section $\sigma_{\alpha} \colon \CP^1 \rightarrow \bL^{\otimes k}$ 
by $\sigma_{\alpha}([\psi]) = \alpha^{\otimes k}|_{\C\psi}$.

\begin{Proposition} \label{Prop:explicitHolRep}
The map $S^k (\C^{2}{}^{*}) \rightarrow \Gamma_{\mathrm{hol}}(\CP^1, \bL^{\otimes k})$ defined by
\[
	\alpha^{\otimes k} \mapsto \sqrt{k+1}\sigma_{\alpha}
\]
is an isometric isomorphism 
of unitary $\SU(2)$-representations. 
\end{Proposition}
\noindent (They are the spin $s = k/2$ representations of dimension $d = k+1 = 2s+1$.)
\begin{proof}
 Since $\bL$ is the dual of the tautological line bundle $\tau$, the linear isomorphism follows from Riemann-Roch.
 The fact that it is an intertwiner follows straight from the definitions, and since both representations are irreducible, 
 it is an isometry up to scaling. The point is to determine the scaling factor.
 
 It is convenient to use the anti-unitary isomorphism $\C^2 \simeq \C^2{}^*$ to write $\alpha(\chi) = \lra{\Omega}{\chi}$ for an $\Omega \in \C^2$.
 Since
$\sigma_{\Omega}([\psi])(l) = \frac{\lra{\Omega}{\psi}^k}{\lra{\psi}{\psi}^k}\lra{\psi^{\otimes k}}{l}$ for all $l \in \tau^{\otimes k}_{[\psi]}$,
we have
\[
h^{\otimes k}(\sigma_{\Omega}, \sigma_{\Omega})_{[\psi]} = \left(\frac{\lra{\Omega}{\psi}\lra{\psi}{\Omega}}{\lra{\psi}{\psi}}\right)^k.
\]
In the projective chart $[1,z] \mapsto z$ around the origin $[1,0]$, the 
curvature of the canonical connection $\nabla$ on $\bL$ is $F_{\nabla} = 1/(1+|z|^2)^2 dz \wedge d\overline{z}$, 
so $F_{\nabla} = -i \omega $ for the symplectic form $\omega = \frac{2}{(1+|z|^2)^2} dx \wedge dy$ with total volume $\vol_{\omega}(\CP^1) = 2\pi$.
With $P_{\psi} = \frac{\ket{\psi}\bra{\psi}}{\lra{\psi}{\psi}}$ and $\Omega = (u, 0)$, we then find
\begin{eqnarray*}
\lra{\sigma_{\Omega}}{\sigma_{\Omega}} &=& \frac{1}{2\pi}\int_{\CP^1} h(\Omega, P_{\psi}\Omega)^k \omega \\
& = & \frac{1}{\pi}\int_{\C} \overline{u}^ku^k(1+|z|^2)^{-2-k}dxdy.
\end{eqnarray*}
The generating function $G(\lambda) = \sum_{k=0}^{\infty} \frac{1}{k!}(-\lambda)^k\lra{\sigma^{(k)}_{\Omega}}{\sigma^{(k)}_{\Omega}}$ is then
\begin{eqnarray*}
G(\lambda) &=& \frac{1}{\pi}\int_{0}^{2\pi}\int_{0}^{\infty} (1+r^2)^{-2}e^{-\lambda |u|^2(1+r^2)^{-1}} rdrd\phi\\
 &=& \int_{0}^{1}e^{-\lambda |u|^2 y} dy = 
\frac{1}{\lambda |u|^2}(1-e^{-\lambda|u|^2}),
\end{eqnarray*}
so $\lra{\sigma^{(k)}_{\Omega}}{\sigma^{(k)}_{\Omega}} = \frac{1}{k+1}\lra{\Omega}{\Omega}^k$. 
\end{proof}

Since $\rho^k$ is a unitary representation, the `conditional expectation' 
\[\Phi_k \colon \cG \rightarrow B(V_k), \quad g \mapsto \Pi_k \circ \rho^k(g) \circ \Pi_k\]
is positive: we have 
\[
	\sum_{i,j} \lra{v_i}{\Phi_k(g^{-1}_{i}g_j)v_j} \geq 0 
\]
for all finite sequences $(v_i, g_i)$ in $V_k \times \cG$. A fortiori, the map $\phi_{k}(g) := \lra{\Omega_k}{\Phi_k(g)\Omega_k}$
is positive for every $\Omega_k \in V_k$. If $\Omega_k$ is cyclic for $(\rho^k, \cH_k)$, then the $\cG$-representation can be recovered from 
$\phi_k$ by the GNS-construction.
Since $\SU(2)$ acts on $\bL$ by \emph{holomorphic} transformations, the 
map $\Phi_k$ is equivariant for both the left and the right action of $\SU(2)$.

Since every holomorphic section is smooth, and since every smooth section is a smooth vector in the $\cG$-representation 
$(\rho^k, \cH_k)$, the positive map $\Phi_k \colon \cG \rightarrow B(V_k)$ is smooth.
Its derivative 
\[
	d\Phi_k(f) = \Pi_k \circ (ikf \one - \nabla^{\otimes k}_{X_f}) \circ \Pi_k
\]
maps into the skew-Hermitian operators, $d\Phi_k \colon C^{\infty}(\CP^1) \rightarrow \fu(V_k)$. 
This map is $\su(2)$-equivariant; in particular, its pullback along the inclusion  $\iota \colon \su(2) \hookrightarrow  C^{\infty}(\CP^1)$ is precisely 
the spin $k+1$ representation $\pi_k \colon \su(2) \rightarrow  \fu(V_k)$.
In other words, the following diagram commutes:
\begin{equation}\label{unlooped}
\begin{tikzcd}[column sep=tiny]
C^{\infty}(\CP^1) \arrow[rr, "d\Phi_k"] &  & \fu(V_k)\\
& \su(2)\arrow[ul, hook, "\iota"]\arrow[ur, hook', "\pi_k"']
\end{tikzcd}
\end{equation}
(Note that $d\Phi_k$ is \emph{not} a Lie algebra homomorphism, but its composition with $\iota$ is.)
If we rescale $d\Phi_k$ by $\hbar := 1/k$, we obtain the \emph{geometric quantisation} map $Q_k := \frac{1}{k} d\Phi_k$,
\[
	Q_{k}(f) = \Pi_k\circ \Big(i f \one - \frac{1}{k}\nabla^{\otimes k}_{X_f}\Big) \circ \Pi_k.
\]
\subsection{Berezin-Toeplitz quantisation}

The \emph{Berezin-Toeplitz quantisation} of $f\in C^{\infty}(\CP^{1}, \R)$ is the operator 
\begin{equation}\label{BerezinToeplitz}
	T_k(f) := \Pi_k \circ (f \one) \circ \Pi_k.
\end{equation}
Note that $T_k \colon C^{\infty}(\CP^1, \R) \rightarrow i \fu(V_k)$ maps into the Hermitian operators.
Berezin-Toeplitz quantisation is related to geometric quantisation by Tuynman's Lemma:
\begin{Lemma}[Tuynman, \cite{Tuynman1987, BHSS91}]
Let $(M, \omega)$ be a prequantisable K\"ahler manifold, with prequantum line bundle $(\bL, \nabla)$.
Then 
\[\Pi_k \circ \nabla_{X_f} \circ \Pi_k = \Pi_k \circ (-i/2) \Delta f \circ \Pi_k,\]
where $\Delta$ is the Laplace-Beltrami operator for the Riemannian metric.
\end{Lemma}
%
Since $\CP^1$ is a K\"ahler manifold, Tuynman's Lemma for $\bL^{\otimes k}$ implies that the geometric quantization $Q_k$ 
is related to the Berezin-Toeplitz quantization $T_k$
by
\[
	Q_k(f) := i T_k(f + {\textstyle \frac{1}{2k}} \Delta f).
\]
In particular, 
this means that $Q_k(f) \rightarrow i T_k(f)$ uniformly in $\|\,\cdot\,\|_{C^2}$ for $k \rightarrow \infty$.

The Berezin-Toeplitz quantisation enjoys the following properties, cf.~\cite{BMS94} Theorem 4.1 and 4.2 for the first two parts,
 and \cite{G79} for the third part.
\begin{Theorem}[\cite{BMS94}]\label{ThmBerezinToeplitzProperties} Let $(\Sigma, \omega)$ be a prequantizable compact K\"ahler manifold of dimension $2d$, and let $\bL \rightarrow \Sigma$
be a  very ample prequantum line bundle. 
\begin{itemize}
\item[1)] For every $f\in C^{\infty}(\Sigma)$, there exists a constant $c_f>0$ such that
\[
\|f\|_{\infty} - \frac{c_f}{k}\leq\|T_k(f)\| \leq \|f\|_{\infty}.
\] 
\item[2)] For every $f, g\in C^{\infty}(\Sigma)$, there exists a constant $c_{f,g}>0$ such that 
\[
	\|k \,[T_k(f), T_k(g)] - iT_k(\{f,g\})\| \leq \frac{c_{f,g}}{k}.
\]
\item[3)] For $f_1, \ldots f_n \in C^{\infty}(\Sigma)$, there exists a constant $c_{f_1, \ldots, f_n} >0$ such that 
\[
\Big| \frac{1}{\mathrm{dim}(V_k)}\tr(T_k(f_1)\cdots T_k(f_n)) - \frac{1}{\vol(\Sigma)}\int_{\Sigma}f_1\cdots f_n\; (\omega^{d}/d!) \Big| \leq \frac{c_{f_1,\ldots,f_n}}{k}.
\]
\end{itemize}
\end{Theorem}
Note that for $\Sigma = \CP^1$, we have $\mathrm{dim}(V_k) = k+1$ and $\vol(\Sigma) = 2\pi$.  

\section{Fuzzy sphere limits of Kac-Moody algebras}\label{Section:limits}

We identify $\fu(V_k)$ with $\fu(k+1)$, and $\su(V_k)$ with $\su(k+1)$.
Geometric quantisation yields a map $d\Phi_k \colon C^{\infty}(\CP^1) \rightarrow\fu(k+1)$,
which reduces to a map $d\overline{\Phi}_k \colon C^{\infty}_0(\CP^1) \rightarrow\su(k+1)$ by
tracing out the constants.
The corresponding map $Ld\overline{\Phi}_k \colon LC^{\infty}_0(\CP^1) \rightarrow L\su(k+1)$
at the level of loop algebras allows us to pull back the Kac-Moody cocycles $\psi_k$ on $L\su(k+1)$
to skew-symmetric bilinear forms $(Ld\overline{\Phi}_k)^*\psi_k$ on $LC^{\infty}_0(\CP^1)$ -- not quite cocycles because 
the map $Ld\overline{\Phi}_k$ is not quite a Lie algebra homomorphism. 
We show that in the limit $k \rightarrow \infty$, these bilinear forms converge to the cocycle $\psi_{\infty}$
from \eqref{eq:cocycleTheoremA}, provided that they are rescaled by a factor of $6/k^{3}$.

\subsection{Comparison of cocycles using geometric quantisation}
On $L\su(k+1)$, every Lie algebra cocycle is cohomologous to $c_k \psi_k$ for some constant $c_k \in \R$, where $\psi_k$ is the 
Kac-Moody cocycle
\[
	\psi_k(X,Y) = \frac{1}{2\pi}\int_{\bS^1}\kappa_{\su(k+1)}(X, {\textstyle\frac{d}{dt}}Y)dt.
\]
Here $\kappa_{\su(k+1)}(X,Y) = \tr(XY)$ is the Killing form of $\su(k+1)$, normalized in such a way 
that the coroots $h = \mathrm{diag}(i,-i,0, \ldots, 0)$ satisfy $\kappa(h, h) = -2$. 
The corresponding central extension integrates to the group level if and only if $c_k$ is an integer \cite[Theorem~ 4.4.1]{PS86}.

It will be convenient to extend the cocycle $\psi_k$ from $L\su(k+1)$ to $L\fu(k+1)$ by the same formula.
We can then pull back the bilinear maps $\psi_k$ to $LC^{\infty}(\CP^1)$ along the map $L(d\Phi_{k}) \colon LC^{\infty}(\CP^1) \rightarrow L\fu(k+1)$. 

With $d\Phi_k(f) = ikT_k(f + \frac{1}{2k}\Delta f)$ and the notation $g' = \frac{d}{dt}g$, we find
\begin{eqnarray*}\textstyle
(Ld\Phi_k)^*\psi_k(f,g)&=& \frac{1}{2\pi} \int_{\bS^1}\kappa_{\su(k+1)}(d\Phi_k(f), d\Phi_k(g')) dt\\
&=& -   \frac{k^2}{2\pi}\int_{\bS^1}\tr(T_k(f + { \textstyle\frac{1}{2k}}\Delta f)T_k(g' + { \textstyle\frac{1}{2k}}\Delta g'))dt\\
&=& - \frac{k^2}{2\pi} \int_{\bS^1} \Big(\frac{k+1}{\vol_{\omega}(\CP^1)}\int_{\CP^1} fg' \omega + \mathcal{O}(1)\Big)dt\\
&=& -  \frac{k^2(k+1)}{2\pi \vol_{\omega}(\CP^1)} \int_{\bS^1}\int_{\CP^1} fg' \omega \wedge dt + \mathcal{O}(k^2).
\end{eqnarray*}
(The last step uses the compatibility between the trace and the 
Liouville measure in Berezin-Toeplitz quantisation, Theorem~\ref{ThmBerezinToeplitzProperties} part 3. The symbol $\mathcal{O}(k^2)$ stands for an additive term that is dominated by $c_{f,g}k^2$, where the constant $c_{f,g}$ is allowed to depend on $f$ as well as $g$.)

Rather than $LC^{\infty}(\CP^1)$ and $L\fu(k+1)$, we would like to compare the perfect Lie algebras $LC^{\infty}_0(\CP^1)$ and $L\su(k+1)$.
Since $d\Phi_k$ does not necessarily map the functions $C^{\infty}_0(\CP^1)$ with zero mean to the trace-free matrices $\su(k+1)$, 
we trace out the constants and define: 
\[d\overline{\Phi}_{k}(f) :=  d\Phi_k(f) - \textstyle{\frac{1}{k+1}}\tr(d\Phi(f)) \one.\]
Pulling back by $Ld\overline{\Phi}_{k}(f)$, we find a correction term 
\[
(Ld\overline{\Phi}_{k})^*\psi_k = (Ld\Phi_{k})^*\psi_k - \textstyle{\frac{1}{2\pi(k+1)}}\int_{\bS^1}\tr(d\Phi_k(f))\tr(d\Phi_k(g'))dt.
\]
Since $\tr(d\Phi_k(f)) = \tr(ikT_k(f) + \mathcal{O}(1))$ and $\tr(T_k(f)) = \frac{(k+1)}{\vol_{\omega}(\CP^1)}\int_{\CP^1}f\omega + \mathcal{O}(1)$, we find
that this correction term is 
\[
+\frac{k^2(k+1)}{2\pi \vol_{\omega}(\CP^1)^2}\int_{\bS^1}\Big(\int_{\CP^1} f\omega\Big)\Big(\int_{\CP^1} g'\mu\Big)dt + \mathcal{O}(k^2).
\]
Denoting the average over $\CP^1$ by \[\overline{f} := \textstyle \frac{1}{\vol_{\omega}(\CP^1)}\int_{\CP^1}f\omega,\] we then find
\begin{equation}\label{eq:ToeplitzLimitCocycle}
(Ld\overline{\Phi}_k)^*\psi_k = - \frac{k^2(k+1)}{2\pi \vol_{\omega}(\CP^1)}\int_{\bS^1}
\int_{\CP^1}(f-\overline{f})(g -\overline{g})' \omega dt + \mathcal{O}(k^2).
\end{equation}
In particular, the correction terms do not affect the highest order term on the subalgebra $C^{\infty}_0(\CP^1)$
of functions that integrate to zero.

Since the cocycle $\psi_{\infty}$ (considered as a cocycle on the Lie algebra of Hamiltonian vector fields) is independent of the choice of volume form, we can choose $\mu = \omega$ to facilitate comparison with $(Ld\overline{\Phi}_k)^*\psi_k$. Since $\vol_{\omega}(\bS^2) = 2\pi$,
the prefactor in \eqref{eq:cocycleTheoremA} becomes $6/(2\pi)^2$, and
we obtain the (pointwise) limit
\begin{eqnarray}\label{limitk3}
\lim_{k \rightarrow \infty} -\frac{6}{k^3} (Ld\overline{\Phi}_k)^*\psi_k = \psi_{\infty}.
\end{eqnarray}


\subsection{Integrality of cocycles and the Limit Theorem}\label{Section:Compatibility}

We consider $\SU(2)$ as a subgroup of $\cG$ by the usual biholomorphic action on the line bundle $\bL \rightarrow \CP^1$, 
and we denote the inclusion by $\iota \colon \SU(2) \hookrightarrow \cG$. 
By the Borel-Weil Theorem, the unitary spin $s=k+1$ representation $\pi_k \colon \SU(2) \rightarrow U(V_k)$ 
is the pullback of 
$\Phi_{k} \colon \cG \rightarrow \mathrm{End}(V_k)$ 
along the inclusion $\iota \colon \SU(2)\hookrightarrow \cG$. Applying the loop functor, we obtain the following commutative diagram:
\begin{equation}\label{CDgrouplevel}
\begin{tikzcd}[column sep=tiny]
\,L\cG\, \arrow[rr, dashed, "L\Phi_k"] &  & L\mathrm{End}(V_k)\\
& LSU(2).\arrow[ul, hook, "L\iota"]\arrow[ur, hook', "L\pi_k"']
\end{tikzcd}
\end{equation}
Note that the (dashed) horizontal map is \emph{not} a group homomorphism, but the two diagonal maps are.
At the Lie algebra level, this yields the commutative diagram 
\begin{equation}\label{CDLieAlgebrawithtrace}
\begin{tikzcd}[column sep=tiny]
LC^{\infty}(\CP^1) \arrow[rr, dashed, "Ld\Phi_k"] &  & L\fu(V_k)\\
& L\su(2),\arrow[ul, hook, "L\iota"]\arrow[ur, hook', "L\pi_k"']
\end{tikzcd}
\end{equation}
where again the map $Ld\Phi_k$ is not a Lie algebra homomorphism.

Since $\su(2)$ maps to 
the space $C^{\infty}_0(\CP^1)$ of functions that integrate to zero, and since
$\overline{d\Phi}_k$ maps into the traceless matrices, we 
obtain the following commutative diagram:
\begin{equation}\label{CDLieAlgebrawithouttrace}
\begin{tikzcd}[column sep=tiny]
LC_0^{\infty}(\CP^1) \arrow[rr, dashed, "L\overline{d\Phi}_k"] &  & L\su(V_k)\\
& L\su(2).\arrow[ul, hook, "L\iota"]\arrow[ur, hook', "L\pi_k"']
\end{tikzcd}
\end{equation}
Again $L\iota$ and $L\pi_k$ are Lie algebra homomorphisms, but the horizontal map 
$L\overline{d\Phi}_k$ is not.
Nonetheless, the cocycles 
\begin{eqnarray}
\psi_{1}(\xi,\eta) &=& \frac{1}{2\pi}\int_{\bS^1}\kappa_{\su(2)}(\xi, \eta') dt\\
\psi_k(X,Y) &=& \frac{1}{2\pi}\int_{\bS^1}\kappa_{\su(k+1)}(X, Y') dt\\
\psi_{\infty}(f,g) &=& \frac{6}{(2\pi)^2}\int_{\bS^1}\Big(\int_{\bS^2}fg' \omega\Big)dt
\end{eqnarray}
on $L\su(2)$, $L\su(k+1)$ and $LC_0^{\infty}(\bS^2)$ can be arranged in a diagram 
\begin{equation}\label{CDcocycleswithouttrace2}
\begin{tikzcd}[column sep=tiny]
c_{\infty} \psi_{\infty} \arrow[dr, two heads, "L\iota^*"']&  & \arrow[dl, two heads, "L\pi^*_k"] \arrow[ll, dashed, "L\overline{d\Phi}^*_k"'] c_k \psi_{k}\\
& c_1 \psi_{1}.
\end{tikzcd}
\end{equation}
For judicious choices of the constants, this diagram is `asymptotically commutative' in the following sense:
\begin{Theorem}\label{Theorem:convergence}
The Lie algebra cocycles $c_1\psi_1$, $c_k\psi_k$ and $c_{\infty}\psi_{\infty}$ integrate to the group level if and only if the 
corresponding constants $c_1$, $c_k$ and $c_{\infty}$ are integers. 
\begin{itemize}
\item[a)]
We have $(L\iota)^* c_{\infty}\psi_{\infty} = c_1 \psi_1$ if and only if 
\[c_{1} = - c_{\infty}.\]
\item[b)]
We have $(L\pi_k)^* c_k \psi_k = c_1 \psi_1$ if and only if 
\[\textstyle c_{1} = \frac{k(k+1)(k+2)}{6} c_k.\]
\item[c)] 
Finally, if $c_k = - \frac{6}{k(k+1)(k+2)}c_\infty$ we have the (pointwise) limit
\[\lim_{k \rightarrow \infty} (Ld\overline{\Phi}_k)^* c_k \psi_k(f,g) = c_{\infty} \psi_{\infty}(f,g).\]
\end{itemize}
\end{Theorem}

\begin{Remark}
Note that $\frac{k(k+1)(k+2)}{6}$ is an integer for all $k\in \Z$, and that the central charges $c_k$ in point c 
satisfy $0 < c_k <1$ for large $k$. In particular, the corresponding cocycles $c_k\psi_k$ do not integrate 
to the group level.
\end{Remark}
%
\begin{proof}
The fact that $c_k\psi_k$ integrates to the group level if and only if $c_k\in \Z$ is well known \cite[Theorem~4.4.1]{PS86}, 
and the fact that $c_{\infty}\psi_{\infty}$ integrates to the group level if and only if $c_{\infty}\in \Z$ is Theorem~\ref{Thm:integralityCocycles}.
Since $\vol_{\omega}(\CP^1) = 2\pi$, the fact that $(L\iota)^*\psi_{\infty} = \psi_1$ if and only if $c_{\infty} = -c_1$ follows from equation~\eqref{eq:Compatible2withInfty} in the proof of Theorem~\ref{Thm:integralityCocycles}.


The pullback of the invariant bilinear form $\kappa_{\su(k+1)}(X, Y) = \tr(XY)$
along the Lie algebra homomorphism $\pi_k \colon \su(2) \rightarrow \su(k+1)$ is an invariant bilinear form
on the simple Lie algebra $\su(2)$, and hence a multiple of $\kappa_{\su(2)}$. 
For part b), it suffices to determine this constant of proportionality.
The Killing form $\kappa_{\su(2)}$ on $\su(2)$ is normalised so that $\kappa_{\su(2)}(h,h) = -2$ for the coroot $h = \mathrm{diag}(i,-i)$.
The Killing form $\kappa_{\su(k+1)}(X, Y) = \tr(XY)$ is normalised so that it evaluates to $-2$ on the coroots of $\su(k+1)$.
Since $\mathrm{spec}(\pi_k(h)) = \{-ik, {-i(k-2),} \ldots, {i(k-2)}, ik\}$, we have 
\begin{eqnarray*}
\kappa_{\su(k+1)}(\pi_k(h), \pi_k(h)) &=& -\big((-k)^2 + (-k + 2)^2 + \ldots + (k-2)^2 + k^2)\\
 &=& -{\textstyle \frac{1}{3}}k(k+1)(k+2).
\end{eqnarray*}
So $\pi^*_k \kappa_{\su(k+1)} = \frac{1}{6}k(k+1)(k+2)\kappa_{\su(2)}$, and 
\[
(L\pi_k)^*\psi_k = {\textstyle\frac{k(k+1)(k+2)}{6}}\psi_1.
\]
Finally, with $c_k = -\frac{6}{k(k+1)(k+2)}c_\infty$, equation \eqref{limitk3} 
yields the pointwise limit $\lim_{k\rightarrow \infty} (L\overline{d\Phi}_k)^* c_k \psi_k(f,g) = c_{\infty}\psi_{\infty}(f,g)$ 
for all $f, g \in C^{\infty}_0(\CP^1)$.
\end{proof}

\subsection{Fuzzy sphere limits for twisted Kac-Moody algebras}\label{sec:twistlim}
The point reflection $P(x) = -x$ on the 2-sphere $\bS^2$ induces an automorphism $\phi \mapsto P \circ \phi \circ P^{-1}$ on $\SDiff(\bS^2, \mu)$, 
and therefore an automorphism $\alpha_{P}(X_f) = P_* X_{f} \circ P^{-1}$ on its Lie algebra $\X(\bS^2, \mu)$. 
Since $P^*\mu = -\mu$, this induces the automorphism $\alpha_{P}(f) = -(P^{-1})^* f$ on the Poisson algebra $C^{\infty}_0(\bS^2)$.

If we identify the 2-sphere with $\CP^1$, then the point reflection  
is the anti-holomorphic map 
\[
P \colon \CP^1 \rightarrow \CP^1, \quad P([z:w]) = [-\overline{w}: \overline{z}].
\]
On the tautological line bundle $\tau \rightarrow \CP^1$, 
it is covered by the anti-holomorphic bundle anti-isomorphism 
\[P^{\tau} \colon \tau \rightarrow \tau, \quad P^{\tau}(z,w) = (-\overline{w}: \overline{z}).
\]
Note that $P^{\tau}$ is antilinear on the fibres, and $(P^{\tau})^2 = -\mathrm{Id}$.
On the dual bundle $\bL = \tau^*$, we then obtain the anti-linear, anti-holomorphic bundle automorphism 
\[
P^{\bL} \colon \bL \rightarrow \bL, \quad P^{\bL}(\alpha) := \overline{\alpha} \circ (P^{\tau})^{-1}.
\]
If $\alpha \in \bL_{[z,w]}$, then $P^{\bL}(\alpha) \in \bL_{P[z:w]}$. Indeed, 
if $\alpha$ is a linear functional on $\tau_{[z:w]}$, then 
$P^{\bL}(\alpha)$ is the concatenation of the antilinear map $\overline{\alpha} \colon \tau_{[z:w]} \rightarrow \C$
with the antilinear map $(P^{\tau})^{-1} \colon \tau_{P[z:w]} \rightarrow \tau_{[z:w]}$, 
and therefore a linear map $P^{\bL}(\alpha) \colon \tau_{P[z:w]} \rightarrow \C$. 
The map $P^{\bL} \colon \bL_{[z:w]} \rightarrow \bL_{P[z:w]}$ is antilinear,  
and $P^{\bL} \colon \bL^{\otimes k} \rightarrow \bL^{\otimes k}$ is an anti-holomorphic 
automorphism that sends the holomorphic section $\sigma_{\Omega}$ from Proposition~\ref{Prop:explicitHolRep}
to the holomorphic section $P^{\bL}_*\sigma_{\Omega} = P^{\bL} \circ \sigma_{\Omega} \circ P^{-1} = \sigma_{(P_*)^{\otimes k} \Omega}$, 
where 
\begin{equation}\label{eq:coordinatesHolSections}
P_* \begin{pmatrix}\Omega_1 \\ \Omega_2\end{pmatrix} = \begin{pmatrix}- \overline{\Omega}_2 \\ \overline{\Omega}_1\end{pmatrix}.
\end{equation}
So $P^{\bL}_* \colon \Gamma_{\mathrm{Hol}}(\bL^{\otimes k}) \rightarrow \Gamma_{\mathrm{Hol}}(\bL^{\otimes k})$ is an anti-linear isomorphism, that can be 
identified with $(P_*)^{\otimes k} \colon S^k \C^2 \rightarrow S^k\C^2$ using Proposition~\ref{Prop:explicitHolRep}.
The map $d\Phi_k \colon C^{\infty}(\bS^2) \rightarrow \fu(V_k)$ as well as its traceless version 
$d\overline{\Phi}_{k} \colon C^{\infty}_0(\bS^2) \rightarrow \su(V_k)$
are equivariant under point reflections.
\begin{Lemma} For $f\in C^{\infty}(\bS^2)$, we have
\[d\Phi_k(\alpha_P(f)) = P^{\bL}_* d\Phi_k(f) (P^{\bL}_*)^{-1}.\]
The same equivariance holds for the traceless version $d\overline{\Phi}_{k}$.
\end{Lemma}
\begin{proof}
Recall that $d\Phi_k(f) = iT_k(kf + \frac{1}{2}\Delta f)$. Since $P^*$ commutes with the Laplace-Beltrami 
operator $\Delta$, 
we have $d\Phi_k(\alpha_{P}f) = iT_k(\alpha_P g)$ for the smooth function $g := kf + \frac{1}{2}\Delta f$, so it suffices to show the  
equivariance of the Berezin-Toeplitz quantisation,  
$iT_k (\alpha_P g) P^{\bL}_*= P^{\bL}_* iT_k(g)$.

Since $P^{\bL}_*$ is antilinear and commutes with the projection $\Pi_k$ on holomorphic sections, 
we have 
\[
P^{\bL}_* iT_k(g) s = - i \Pi_k P^{\bL}_* (g s) = -i \Pi_k((P^{-1})^* g) P^{\bL}_*s
\]
for holomorphic sections $s$ of $\bL^{\otimes k}$, which coincides with 
\[
iT_k (\alpha_P g) P^{\bL}_*(s)  =  i \Pi_k (- P^{-1})^* g) P^{\bL}_*s
\]
as required.
From this we readily find the equivariance of the traceless version:
\begin{eqnarray*}
d\oline{\Phi}_k(\alpha_P(f)) & = & d\Phi_k(\alpha_P(f)) - \textstyle{\frac{1}{k+1}} \tr\big(d\Phi_k(\alpha_P(f))\big) \one\\
& = & P^{\bL}_* d\Phi_k(f) (P^{\bL}_*)^{-1} - \textstyle{\frac{1}{k+1}} \tr\big(P^{\bL}_* d\Phi_k(f) (P^{\bL}_*)^{-1}\big) \one\\
& = & P^{\bL}_* \Big( d\Phi_k(f) (P^{\bL}_*)^{-1} - \textstyle{\frac{1}{k+1}} \tr\big(d\Phi_k(f) \big) \one\Big) (P^{\bL}_*)^{-1}\\
& = & P^{\bL}_* d\oline{\Phi}_k(\alpha_P(f)) (P^{\bL}_*)^{-1}. 
\end{eqnarray*}%
\end{proof}

Since $\mathfrak{su}(2)$ is invariant under point reflections, the commutative diagram 
\eqref{CDLieAlgebrawithouttrace} restricts to the commutative diagram
\begin{equation}\label{CDLieAlgebrawithouttraceReflectionInvariant}
\begin{tikzcd}[column sep=tiny]
LC_0^{\infty}(\CP^1)^{P} \arrow[rr, dashed, "L\overline{d\Phi}_k"] &  & L\su(V_k)^{P}\\
& L\su(2).\arrow[ul, hook, "L\iota"]\arrow[ur, hook', "L\pi_k"']
\end{tikzcd}
\end{equation}
Here $LC_0^{\infty}(\CP^1)^{P}$ denotes the twisted loop algebra of all smooth functions 
$F \colon \R \rightarrow C^{\infty}_0(\bS^2)$ that satisfy $F_{t+1} = \alpha_P (F_t)$, and 
$L\su(V_k)^{P}$ denotes the twisted loop algebra of all smooth functions 
$\xi \colon \bS^1 \rightarrow \su(V_k)$
with $\xi_{t+1} = P^{\bL}_* \xi_t (P^{\bL}_*)^{-1}$.

If we use Proposition~\ref{Prop:explicitHolRep} to identify $\Gamma_{\mathrm{hol}}(\bL^{\otimes k})$  
with the $\mathrm{SU}(2)$-representation $S^k\C^2$, then by \eqref{eq:coordinatesHolSections}, the 
anti-linear map $P_*^{\otimes k} \colon S^k\C^2 \rightarrow S^k\C^2$ is given by 
$P_*^{\otimes k} = \pi_k(u)\circ J = J \circ \pi_k(u)$, where $J$ is complex conjugation on $S^k\C^2$ and 
\[
u  = \exp(-\textstyle{\frac{1}{2}}\pi i \sigma_y) = \begin{pmatrix} 0&-1\\1&0\end{pmatrix}.
\]
One readily checks that conjugation of $\su(2)$ by $P_*$ is trivial, so $L\su(2)^{P} = L\su(2)$.
For $k>1$, the conjugation by $\pi_k(u)$ does not affect the 
isomorphism class of $L\su(k+1)$ because it is an inner automorphism for the connected Lie group $\mathrm{SU}(k+1)$, 
but for $k>1$ conjugation by the antilinear map $J$ is not an inner automorphism: it is an antilinear intertwiner between 
the defining representation and its dual.

For $k>1$, then, the Lie algebra $L\su(V_k)^{P}$ is not isomorphic to the loop algebra $L\su(V_k)$, but
its complexification $L\su(V_k)^{P}_{\C}$ is the twisted loop algebra corresponding to the unique nontrivial 
diagram automorphism of $\mathfrak{sl}(k+1) \simeq A_k$.

We identify $LC^{\infty}_0(\CP^1)^{P}$ with $C^{\infty}_0(\bS^1 \times \bS^2)^{\Z_2}$ as before, 
and if we identify $L\su(V_k)^{P}$ with the fixed point algebra $C^{\infty}(\bS^1, \su(k+1))^{\Z_2}$ 
under conjugation by $\pi_k(u)J$. Then the 
cocycles on $L\su(2)$, $L\su(k+1)^{\Z_2}$ and $C^{\infty}_0(\bS^1 \times \bS^2)^{\Z_2}$ given by
\begin{eqnarray}
\psi_{1}(\xi,\eta) &=& \frac{1}{2\pi}\int_{\bS^1}\kappa_{\su(2)}(\xi, \eta') dt\\
\psi^{P}_k(X,Y) &=& \frac{1}{4\pi}\int_{\bS^1}\kappa_{\su(k+1)}(X, Y') dt\\
\psi_{\infty}^{P}(f,g) &=& \frac{3}{(2\pi)^2}\int_{\bS^1}\Big(\int_{\bS^2}fg' \omega\Big)dt
\end{eqnarray}
can again be arranged in a diagram 
\begin{equation}\label{CDLieAlgebrawithouttraceTwist}
\begin{tikzcd}[column sep=tiny]
c^P_{\infty} \psi^{P}_{\infty} \arrow[dr, two heads, "L\iota^*"']&  & \arrow[dl, two heads, "L\pi^*_k"] \arrow[ll, dashed, "L\overline{d\Phi}^*_k"'] c^P_k \psi^P_{k}\\
& c_1 \psi_{1}.
\end{tikzcd}
\end{equation}

The prefactors were chosen in such a way that integral multiples correspond to integrable cocycles;
see Theorem~\ref{ThmCocycleTwistedLoop} for $C^{\infty}_{0}(\bS^1 \times \bS^2)^{\Z_2}$, the argument for $C^{\infty}(\bS^1, \su(k+1))$
is similar.
By restricting Theorem~\ref{Theorem:convergence} to the fixed point algebras, we immediately find the corresponding result in the twisted case:
\begin{Theorem}\label{Theorem:convergenceTwist}
The Lie algebra cocycles $c_1\psi_1$, $c^P_k\psi^P_k$ and $c^P_{\infty}\psi^P_{\infty}$ integrate to the group level if and only if the 
corresponding constants $c_1$, $c^P_k$ and $c^P_{\infty}$ are integers. 
\begin{itemize}
\item[a)]
We have $(L\iota)^* c^P_{\infty}\psi^P_{\infty} = c_1 \psi_1$ if and only if 
\[c_{1} = - \textstyle{\frac{1}{2}}c^{P}_{\infty}.\]
\item[b)]
We have $(L\pi_k)^* c^P_k \psi_k = c_1 \psi_1$ if and only if 
\[\textstyle c_{1} = \frac{k(k+1)(k+2)}{12} c^{P}_k.\]
\item[c)] 
Finally, if $c^{P}_k = - \frac{6}{k(k+1)(k+2)}c^{P}_\infty$ we have the (pointwise) limit
\[\lim_{k \rightarrow \infty} (Ld\overline{\Phi}_k)^* c^{P}_k \psi_k(f,g) = c_{\infty} \psi^{P}_{\infty}(f,g).\]
\end{itemize}
\end{Theorem}
\section{Beyond the 2-sphere}\label{sec:BeyondSpheres}

For definiteness we have restricted attention to the 2-sphere $\bS^2$, but it is natural to ask what happens 
for more general symplectic manifolds $(\Sigma, \omega)$. For the classification of cocycles, the natural 
starting point would be $C^{\infty}_{c}(M, \fk)$ with $\fk = C^{\infty}_{c, 0}(\Sigma)$ the Poisson Lie algebra
of compactly supported functions that integrate to zero. The `twisted' version can then be described 
in terms of an (unoriented) symplectic bundle $F \rightarrow M$ with typical fibre $(\Sigma, |\omega|)$.

Although we expect that the above approach to central extensions can be adapted to this setting,
we have not carried out the analysis in detail. At this point we would like to briefly point out some ideas, conjectures and 
open problems in this direction. 

\subsection{Conjectural classification of 2-cocycles}
For a (not necessarily compact) symplectic manifold $(\Sigma, \omega)$,  
the Poisson Lie algebra $\fk = C^{\infty}_{c, 0}(\Sigma)$ 
is perfect 
\cite{ALDM74, JV16}. One could attempt a localization procedure along the lines of \cite{JW13},
and if this works, one can restrict attention to the local case $\fg_{c}(M) = C^{\infty}_{c}(M, \fk)$.

Then \cite{NW08} suggests that one should expect two sources of cocycles: the invariant bilinear forms on $\fk$ and the 
second Lie algebra cohomology $H^2(\fk, \R)$. By \cite{JV16}, the second continuous Lie algebra cohomology
of $\fk$ is isomorphic to $H^1_{\mathrm{dR}}(M)$. A symplectic vector field $S$ on $\Sigma$
gives rise to the 2-cocycle $\psi(f,g) = \int_{\Sigma} (fL_{S}g - gL_{Sf})\omega^n/n!$, and this 2-cocycle 
is exact if and only if $S$ is Hamiltonian.
From this, one obtains 2-cocycles on $\fg_{c}(M)$ of the form 
\begin{equation}\label{eq:ExtraH2}
\psi_{\phi, S}(F, G) = \phi\Big(\int_{\Sigma} (F L_{S} G - G L_{S}F) \;\omega^{n}/n!\Big),
\end{equation}
where $S \colon M \rightarrow \X(\Sigma,\omega)$ is a smooth function with values in the symplectic 
vector fields, and $\phi$ is a distribution on $M$. They do not show up for $\bS^2$ because all symplectic 
vector fields on $\bS^2$ are hamiltonian (so they give rise to coboundaries), so these `vertical cocycles'
are an additional feature of the problem. 

The second source of cocycles to expect is the space of continuous invariant bilinear forms. One continuous invariant bilinear form on $C^{\infty}_{c, 0}(\Sigma)$ is 
the usual inner product \[\kappa(f,g) = \int_{\Sigma}fg \omega^n/n!,\] 
but as far as we are aware it is an open problem to prove that it is unique up to scaling.
One possible approach would be to first prove it for a suitable local model (e.g.~$T^*\R^n$, or $\CP^n$), 
and then use arguments in the spirit of \cite[Lemma~2.2]{JW13} to localize the problem by showing that every continuous 
invariant bilinear form on the perfect Lie algebra $\fk$ arises from a distribution on $\Sigma \times \Sigma$
that is supported on the diagonal $\Delta(\Sigma)$.

An interesting new feature compared to the case of $\bS^2$ is that for a (necessarily noncompact) symplectic manifold $\Sigma$
with a conformal symplectic vector field $E$ satisfying $L_{E} \omega = c \omega$, the Koszul map 
$\Gamma(\kappa)(f, g, h) = \kappa(\{f, g\}, h\})$ is \emph{not} injective, at least if $c\neq 0$. 
Indeed, if $X_f$ is a hamiltonian vector field with potential $f$, then $[E, X_f]$ is a hamiltonian vector field 
with potential $L_{E}f - cf$. With this, one readily sees that  
\[
L_{E}\{f, g\} 
= \{L_{E}f, g\} + \{f, L_{E}g\} - c\{f,g\},
\]
and that the 2-cochain $\eta(f,g) = \int_{\Sigma}(fL_{E}g - gL_{E}f)\omega^n/n!$ satisfies
\begin{eqnarray*}
d\eta(f,g,h) 
& = & (n+2)c \int_{\Sigma} \{f,g\}h \omega^n/n!.
\end{eqnarray*}
Following \cite{NW08}, we then obtain for \emph{any} (not necessarily closed) current $\lambda \in \Omega_{c}^1(M)'$
the \emph{coupled cocycle} 
\begin{equation}\label{eq:ExtracocycleKoszul}
\psi_{\lambda}(F,G) = \lambda \Big(\int_{\Sigma} (n+2) c \Big(Fd^{M}G - Gd^{M}F\Big) - d^{M}(FL_{E}G - GL_{E}F) \omega^n/n!\Big), 
\end{equation}
where $d^{M}$ denotes the de Rham differential on $M$.
Note that if $\lambda$ is closed, $\lambda(d\Omega_{c}^0(M)) = 0$, then this reduces up to scaling to the usual cocycle 
\begin{equation}\label{UsualCocycle}
\psi_{\lambda}(F,G) = \lambda\Big(\int_{\Sigma} \big(Fd^{M}g - Gd^{M}F\big)\omega^n/n!\Big)
\end{equation}
that we expect for compact manifolds. For symplectic vector fields $S$ (conformal vector fields $E$ with $c=0$), one recovers the cocycles from \eqref{eq:ExtraH2} 
that come from an exact current $\phi \in \Omega^0_{c}(M)'$, in the sense that $\phi = \lambda \circ d^M$.
 
Conjecturally, then, the picture would be that the 2-sphere $\Sigma = \CP^1$ admits only for the cocycles 
\eqref{UsualCocycle} coming from the invariant bilinear form.
A single puncture $\Sigma = \CP^1 \setminus \{z_0\}$ would result in additional cocycles 
of the coupled type \eqref{eq:ExtracocycleKoszul}, associated to the Euler vector field $E$ on $\CP^1\setminus \{z_0\} = T^*\R$.  
For every additional puncture $\Sigma = \CP^1 \setminus \{z_0, \ldots, z_n\}$, one
expects additional cocycles of `vertical' type \eqref{eq:ExtraH2}, associated to symplectic vector fields $S_n$
whose class in $H^1_{\mathrm{dR}}(\Sigma)$ is Poincar\'e dual to a loop around the extra puncture $z_n$. 

The cocycles where the integrability is best understood are probably those of general type \eqref{UsualCocycle}, where 
results from \cite{NW09} imply that integrability hinges on a period homomorphism $\pi_2(\mathrm{Ham}_{c}(\Sigma)) \rightarrow \R$
for the (simply connected cover of the) group of compactly supported hamiltonian diffeomorphisms with compactly supported potentials. 

For the `coupled' and `vertical' cocycles, the precise integrability conditions are an open problem. 
For loop groups, one possible approach could be to evaluate the period homomorphism on $\pi_2(L\mathrm{Ham}_c(\Sigma))$
provided in \cite{Ne02}. Alternatively, for the `vertical' cocycles of type \eqref{eq:ExtraH2},
the ideas from \cite{JV18} and \cite{DJNV21, DJNV24} might be useful 
to obtain at least sufficient conditions for integration. 

\subsection{Conjectural limits for compact K\"ahler manifolds}

For the `fuzzy sphere limits' the natural setting would be the loop algebra $L\fk$, where 
$\fk = C^{\infty}_{0}(\Sigma)$ is the Poisson algebra for a compact K\"ahler manifold $\Sigma$.
Needless to say, geometric quantization and Berezin-Toeplitz quantization are well adapted to this setting.
If $\bL \rightarrow \Sigma$ is a prequantum line bundle, then 
geometric quantization yields a map from $C^{\infty}_{0}(\Sigma)$ to $\mathrm{End}(\Gamma_{\mathrm{hol}}(\bL^{\otimes k}))$.
Since the space of holomorphic sections is finite dimensional, one can apply the loop functor to obtain natural 
maps from $L\fk$ to $L\su(d(k))$, where $d(k)$ is the dimension of the space of holomorphic sections of $\bL^{\otimes k}$ 
\cite{BHSS91,BMS94}.
Although we have carried out a detailed analysis only for $\Sigma = \CP^1$, it seems likely that a similar procedure 
should yield interesting `geometric quantization limits' for cocycles of the type \eqref{UsualCocycle}.

For the cocycles of `vertical' and `coupled' type, it would be interesting to investigate if -- and if so, how -- these arise 
from (possibly trivial) cocycles on loop algebras with finite dimensional fibres in the limit $k \rightarrow \infty$.

 \section*{Acknowledgments}
Z.W. is partially supported by ARO MURI contract W911NF-20-1-0082.
 
\bibliographystyle{alpha}

\end{document}